\documentclass[sigconf]{acmart}

\AtBeginDocument{%
  }

\copyrightyear{2026}
\acmYear{2026}
\setcopyright{cc}
\setcctype{by}
\acmConference[CCS '26]{Proceedings of the 2026 ACM SIGSAC Conference on Computer and Communications Security}{November 15--19, 2026}{The Hague, Netherlands}
\acmBooktitle{Proceedings of the 2026 ACM SIGSAC Conference on Computer and Communications Security (CCS '26), November 15--19, 2026, The Hague, Netherlands}
\acmDOI{10.1145/3830454.3832583}
\acmISBN{979-8-4007-2871-6/2026/11}
\usepackage{graphicx} 
\usepackage{booktabs}
\usepackage{amsmath}

\usepackage[caption=false,font=footnotesize,labelfont=rm,textfont=rm]{subfig}

\usepackage{amsthm}
\usepackage{cleveref}
\usepackage{multirow}
\usepackage{algorithm}
\usepackage{algpseudocode}

\usepackage{xcolor}

\newcommand{\TLSS}{TLSS\xspace}
\newcommand{\NFSA}{NFSA\xspace}

\crefname{table}{Table}{Tables}
\crefname{definition}{Def.}{Defs.}
\crefname{theorem}{Thm.}{Thms.}
\crefname{equation}{Eq.}{Eq.s}
\crefname{figure}{Fig.}{Figs.}
\crefname{appendix}{Appendix}{Appendices}
\crefname{property}{Property}{Properties}
\crefname{lemma}{Lemma}{Lemmas}
\crefname{algorithm}{Alg.}{Algs.}

\usepackage{xspace}
\newtheorem{theorem}{Theorem}
\newtheorem{lemma}{Lemma}

\begin{document}

\title{NFSA: Non-Forward Secure Aggregation with One Server via Two Layer Secret Sharing}


\author{Yufei Zhou}
\email{zhouyf55@mail2.sysu.edu.cn}
\orcid{0009-0000-3494-7663}
\affiliation{%
  \institution{Sun Yat-sen University}
  \city{Guangzhou}
  \country{China}
}

\begin{abstract}
Federated Learning (FL) enables collaborative model training while preserving privacy by keeping data local. However, the risk of sensitive data leakage through model updates necessitates the use of secure aggregation protocols. Existing server-based secure aggregation protocols typically require the server to forward sensitive data shared between users, which increases communication overhead and introduces potential security risks.
In this work, we propose a novel secure aggregation protocol based on two-layer secret sharing to address these issues. By combining Shamir's Secret Sharing with 2-out-of-2 additive secret sharing using a Pseudo-Random Function (PRF), our protocol eliminates direct communication between users, thereby removing the need for the server to forward data.
We further extend the protocol with Key-homomorphic PRF (KhPRF)  to support high-dimensional data aggregation and apply it to FL, enabling one-shot secure aggregation with a single server and no intermediary data forwarding. To reduce user overhead, we design a new encoding method based on the Chinese Remainder Theorem for the almost KhPRF-based mask, reducing the number of KhPRF calls and mitigating the model update expansion issue after masking. Experimental results show that our scheme significantly outperforms existing methods in terms of auxiliary node overhead. For instance, when the number of users is 100, our scheme improves communication efficiency by nearly 100 times and reduces computational overhead by approximately 17\%. Moreover, user computation time can be reduced by 51\% to 75\% when the input length is $2^{18}$.
\end{abstract}

\ccsdesc[500]{Security and privacy~Cryptography}
\ccsdesc[500]{Security and privacy~Privacy-preserving protocols}
\ccsdesc[500]{Security and privacy~Distributed systems security}

\keywords{Secure Aggregation; Secret Sharing; Key-homomorphic Pseudo-Random Function}

\maketitle

\section{Introduction}

Federated Learning (FL) \cite{mcmahan2017communication} enhances user privacy by ensuring that the server only knows the aggregated model, not the raw data. This feature has led to widespread research and applications in areas such as healthcare \cite{xu2021federated}, financial services \cite{oualid2025federated}, and intelligent manufacturing \cite{li2024novel}. However, despite not directly uploading private data, users' privacy may still be compromised through plain model gradients (or model parameters) uploaded by users \cite{zhu2019deep,ren2022grnn}. 
To address this issue, secure aggregation protocols \cite{bonawitz2017practical} have been proposed. Secure aggregation employs cryptographic techniques, such as random masking \cite{bonawitz2017practical,bell2020secure} and homomorphic encryption \cite{aono2017privacy,zhou2024two}, to ensure that model aggregation can be performed without disclosing users' private model gradients.

Designing efficient secure aggregation schemes requires addressing two key issues, particularly in cross-device scenarios \cite{yang2019federated}. \textit{First, the problem of user dropout}. In cross-device FL, users may be mobile devices \cite{lee2024federated} or resource-constrained IoT devices \cite{nguyen2021federated}, whose connectivity is often unstable. As a result, users may drop out at any time. Secure aggregation in such scenarios must ensure that the privacy of other users is not compromised and that training can continue even if some users drop out.
\textit{Second, the issue of communication efficiency}. Both the model parameters and the number of users are large. However, users' communication capabilities are often limited, typically with low transmission rates. Consequently, communication overhead frequently becomes a bottleneck for training efficiency.

For the first issue (user dropout), multiple-server solutions \cite{fazli2023fedshare,liang2024privacy} can easily resolve the problem. However, they rely on the assumption that servers do not collude, which is often not feasible in many scenarios. Moreover, deploying multiple servers increases the cost of practical deployment. Therefore, we focus on the single-server scenario.
Single-server schemes typically address this issue using threshold secret sharing (SS), such as Shamir’s SS \cite{shamir1979share}. A typical example is given in \cite{bonawitz2017practical}, where user $U_i$ shares its private key with other users (the holders) through SS. When $U_i$ drops out, the server can reconstruct $U_i$'s private key using the secret shares returned by the other users, thereby mitigating the impact of the dropout.
Subsequent work has focused on improving the efficiency of private key sharing, particularly by reducing the number of holders, which in turn decreases sharing overhead. For instance, \cite{bell2020secure} reduces the number of holders by constructing a neighbor graph. Other works, such as \cite{li2023lerna, bell2025willow, karthikeyan2025opa}, elect a committee as holders to hold the private key, further reducing the number of holders.

However, a core issue remains unresolved: the sharing still requires forwarding through an untrusted server. There is no direct communication channel between users or committees, so the communication must be relayed through the server. Specifically, user $U_i$ needs to send the share $s_m$ to the server, which then forwards $s_m$ to the share holder $P_m$. Since we do not fully trust the server, $U_i$ must employ additional cryptographic techniques, such as authenticated encryption (AE) \cite{jimale2022authenticated}, to ensure that the server cannot tamper with the share and forward it to the correct recipient. Furthermore, the server should not learn any information about the forwarded share.
In addition to security, server forwarding introduces significant communication overhead. With $N$ users and $M$ share holders, the forwarding complexity is $O(NM)$. If there are $R$ aggregation rounds, the total communication complexity becomes $O(NMR)$.

To address this issue, we propose a two-layer SS scheme (\TLSS) that, combined with the properties of Pseudo-Random Functions (PRF) \cite{naor2004number}, enables secure aggregation without the need for server relaying. The first layer uses threshold SS, such as Shamir’s scheme \cite{shamir1979share}, to prevent user dropouts. The second layer uses 2-out-of-2 additive SS \cite{krenn2023introduction} to share the Shamir share $s_m$ with both the server and the share holder $P_m$. By using 2-out-of-2 additive SS with a PRF and a pre-negotiated key $K_m$, we can achieve secure aggregation without any server forwarding.
It is important to note that a KA protocol, such as Diffie-Hellman (DH) \cite{harn2014efficient}, is done by exchanging public keys and does not require the server to relay any secret information. 

In \Cref{fig:ss_compare}, we present a communication diagram comparing our \TLSS with traditional Shamir’s SS in the FL scenario for sharing and reconstruction. $s_m^{AE}$ represents the share value encrypted with AE, which needs to be sent to the share holder $P_m$. $s_m^{A_1}$ and $s_m^{A_2}$ represent the additive secret shares satisfying $s_m = s_m^{A_1} + s_m^{A_2} \mod p$. Our scheme reduces the amount of forwarded communication for sharing and does not require AE encryption.

\begin{figure}
    \centering
    \includegraphics[width=\linewidth]{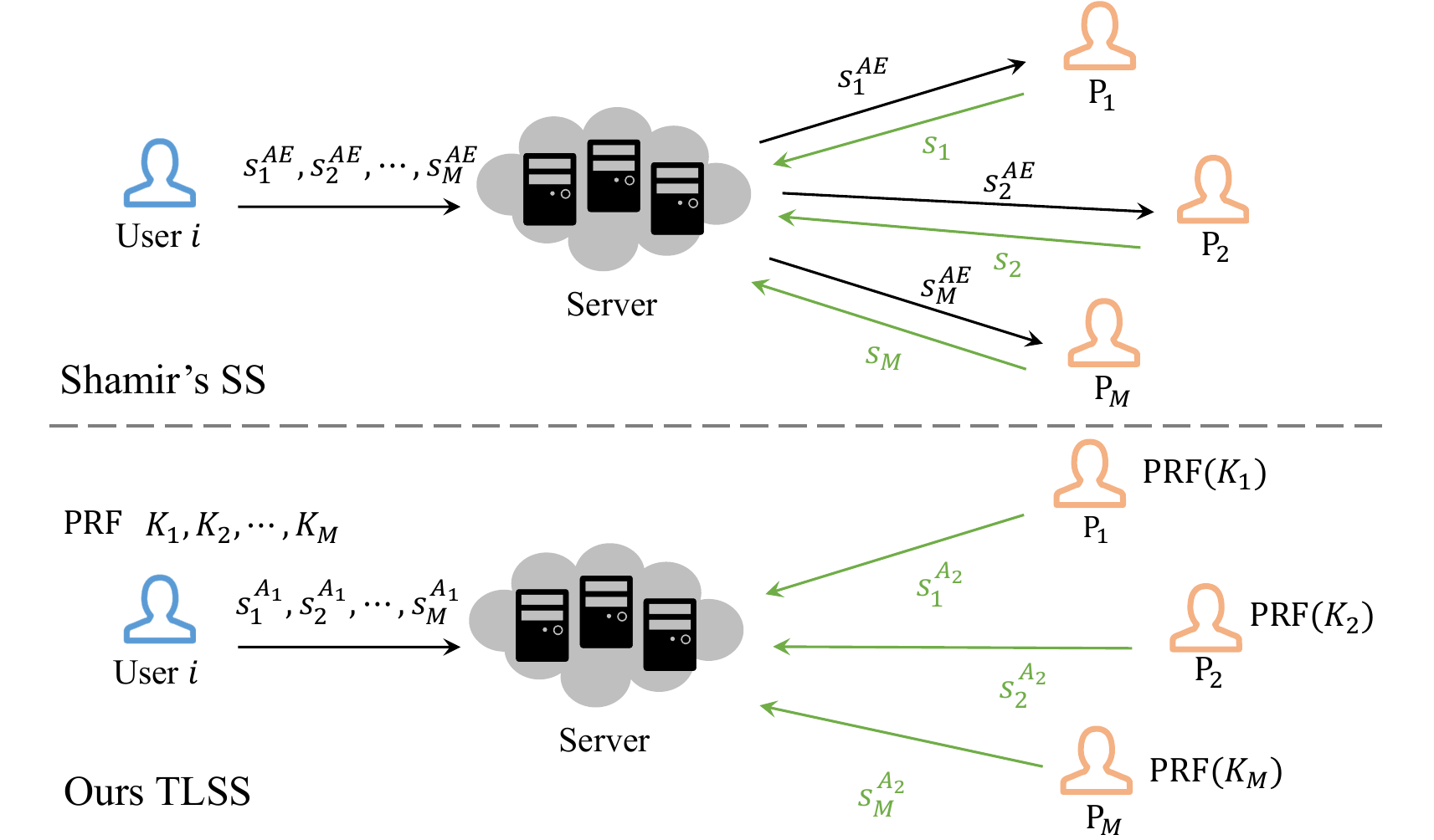}
    \caption{Comparison between our \TLSS and traditional Shamir’s SS. The superscript AE indicates AE encryption, while A denotes additive SS. The black arrows represent the sharing process, and the green arrows represent the reconstruction process.}
    \label{fig:ss_compare}
\end{figure}

For the second issue (communication efficiency), secure aggregation must minimize user communication as much as possible. This reduction in communication involves two aspects: one is reducing the number of communication rounds, and the other is minimizing the total communication volume. Recent works \cite{karthikeyan2025opa, bell2025willow, guan2025opsa} further decrease the number of interactions, enabling one-shot aggregation schemes where users and the server interact only once per round.
The key technology used is an almost key-homomorphic PRF (KhPRF) based on LWR or LWE assumptions \cite{boneh2013key, ernst2021private}. However, almost KhPRF introduces noise, and the noise size is proportional to the number of users participating in the aggregation. Thus, prior works must allocate extra space during model update masking to avoid interference from KhPRF’s noise in the aggregation result. Undoubtedly, this increases the total communication volume. The additional communication per user per round is approximately $O(\log N)$, so the total additional communication for $N$ users and $R$ rounds of aggregation is $O(RN\log N)$.

In this work, we propose a new encoding method (CRT encoding), which utilizes the Chinese Remainder Theorem  to reduce the ciphertext expansion ratio (the ratio of the size of the masked input to the size of the original input), thereby decreasing the total communication volume for users. Additionally, the encoding reduces the number of KhPRF calls by combining multiple inputs into a single one, thereby lowering the computational overhead. Furthermore, we combine the proposed \TLSS with KhPRF and achieve a more compact, one-shot secure aggregation (\NFSA) that does not require server relaying through round fusion. Notably, in one round of secure aggregation, our scheme requires communication only once between the users involved in the aggregation and the server, and the decryptors also need to communicate with the server only once. Our main contributions are summarized as follows:

\begin{itemize}
\item We propose \TLSS, which addresses the issue of forwarding secret shares in the single-server FL scenario, as encountered in prior work using Shamir's SS.
\item We analyze the error bounds of almost KhPRF and introduce a new encoding method that achieves more compact masking of model updates, reducing both the computational overhead and communication volume required by users.
\item We combine the proposed \TLSS with the new encoding for almost KhPRF to enable compact, one-shot aggregation (\NFSA) for high-dimensional model parameter aggregation, without the need for server relaying.
\end{itemize}

\section{Related Work}
In this section, we briefly review related work on cross-device secure FL in a single-server setting.

Bonawitz et al.~\cite{bonawitz2017practical} proposed the first double-mask protocol for FL, where each pair of users shares a symmetric key to ensure the privacy of the aggregated model. Bell et al.~\cite{bell2020secure} later enhanced this approach by improving the neighbor graph construction to reduce communication costs. FLDP~\cite{stevens2022efficient} and ACORN~\cite{bell2023acorn} further optimize the one-time masking process by leveraging the additive homomorphism of (Ring) LWE \cite{rosca2018ring}. However, these schemes rely on one-time keys, requiring clients to regenerate and distribute fresh keys in each round, which results in significant communication interactions.

To address this issue, MicroFedML~\cite{guo2022microfedml} introduces a one-time setup that allows for key reuse across rounds, effectively reducing communication costs. Flamingo~\cite{ma2023flamingo} combines threshold decryption with double masking, while LERNA~\cite{li2023lerna} employs flat SS and a KhPRF to achieve aggregation in a single round. However, these schemes are not suitable for high-dimensional data aggregation. OPA~\cite{karthikeyan2025opa} proposes an almost KhPRF based on LWE assumptions, and Willow~\cite{bell2025willow} presents a symmetric encryption scheme that is both key- and message-homomorphic to aggregate data with one-shot user communication.

To ensure users do not drop out, most of these schemes use Shamir's SS. Although these works have made significant progress in achieving efficient secure aggregation, they do not address the issue of secret message forwarding by the server and the communication overhead caused by user communication expansion.
We propose \TLSS to address the forwarding issue and design a new input encoding method for almost KhPRF that reduces the overhead of secure aggregation. 
The following is a comparative analysis between our approach and the state-of-the-art approaches.

\textbf{Comparison with OPA.}
Compared to OPA, both our scheme and OPA employ Shamir's SS to share the KhPRF key with the decryptors and use the KhPRF to mask inputs. The decryptors then reconstruct the aggregated KhPRF key to assist the server in decoding. However, OPA uses symmetric encryption to encrypt the shared values before forwarding them to the decryptors. In contrast, our scheme employs \TLSS and PRF to avoid forwarding of the shares, thus reducing the overhead of key sharing. Additionally, the proposed CRT encoding reduces the input dimension, lowering both the computational and communication costs of KhPRF. OPA primarily focuses on verifying the correctness of key sharing and KhPRF masking using SCRAPE Test \cite{cascudo2024publicly} and ZKP techniques. In contrast, our scheme aims to mitigate the security risks caused by key forwarding and enhance the performance of KhPRF through CRT packing in the semi-honest model. As a result, our scheme achieves better efficiency than OPA under semi-honest conditions, whereas OPA provides a better global model in a malicious environment due to its verification mechanisms.

\textbf{Comparison with LERNA.}
Although LERNA also uses a KhPRF and requires only a single round of interaction, it employs a flat SS instead of Shamir's SS. LERNA requires every client to communicate with the decryptors during the setup phase, leading to a significant setup overhead, with decryptors needing several gigabytes of communication. The large setup cost makes LERNA less suitable for FL. In contrast, our scheme has stateless decryptors who only need to download the users' public keys from the server, without any additional setup or communication overhead, and supports dynamic changes.

\textbf{Comparison with Willow.}
Our optimizations for almost KhPRF and SS are specifically tailored for OPA-like frameworks and are not directly applicable to Willow’s architecture. In Willow, the subprotocol to aggregate keys is based on an asymmetric message homomorphic scheme, while our approach relies on a SS-based  method. On one hand, our scheme requires only a single round of participation from stateless decryptors. In contrast, Willow involves two rounds of interaction and relies on stateful decryptors. While Willow focuses on the single-iteration setting, it is unclear how it can be extended to a multi-iteration setting without requiring a two-round setup procedure. As a result, Willow is not suitable for multi-round secure aggregation, and its scheme is more appropriate for single-shot histogram aggregation over large client populations while our method focuses on  secure aggregation in FL.

\section{Preliminaries}

\subsection{Notations}
\label{sec:notations}
For convenience in the following descriptions, we define some notations here. Let $N$ represent the number of users, $M$ the number of share holders, $R$ the number of aggregation rounds, and $d$ the dimension of the model parameters. We use bold lowercase letters to represent vectors, e.g., $\mathbf{x}$ denotes a vector, and $\mathbf{x}[i]$ represents the $i$-th element. We use $\mathcal{D}_i$ to denote the dataset of user $U_i$. The size of the dataset $\mathcal{D}_i$ is represented as $\Vert \mathcal{D}_i \Vert$. We use $\mathbb{Z}_p$ to denote the integer field, with elements ${0, 1, \dots, p-1}$, where $p$ is a prime number. All operations over $\mathbb{Z}_p$ are performed modulo $p$. We use $\{x_i\}_{i \in [1:n]}$ to denote the set $\{x_1, x_2, \dots, x_n\}$ obtained by iterating over all $i \in \{1, 2, \dots, n\}$.

\subsection{Secure Aggregation for Federated Learning}
We assume that the input $\mathbf{x}_i^{(r)}$ for each user $U_i$ in the $r$-th round of secure aggregation is a $d$-dimensional vector in the ring $\mathbb{Z}_{p_m}$. The goal of secure aggregation is to securely compute:
\begin{equation}
\label{eq:fl_aggregation}
\mathbf{x}_0^{(r)} = \sum_{i \in \mathcal{U}} \mathbf{x}_i^{(r)} \bmod p_m,
\end{equation}
where $\mathcal{U}$ denotes the set of users online in the $r$-th round.
More details about secure aggregation and FL are in \cref{apx:pre_sa}.

\subsection{Secret Sharing}
In a Secret Sharing (SS) scheme, there exists a dealer and $M$ share holders, denoted as $P_1, P_2, \cdots, P_M$. The dealer divides the secret $s$ into $M$ shares, $s_1, s_2, \cdots, s_M$, using a sharing algorithm, and transmits each share $s_m$ to the corresponding holder $P_m$ via a secure communication channel. The security of the sharing process ensures that no single share $s_m$ reveals any information about the secret $s$. 
When reconstruction is needed, each holder $P_m$ sends their share $s_m$ to the reconstruction entity (the server). The server then uses a reconstruction algorithm along with the shares $s_1, s_2, \cdots, s_M$ to recover the secret $s$.
The details of the used 2-out-of-2 additive SS and Shamir's SS in this work are described in \cref{apx:ss}.

\subsection{Pseudo-Random Function}
A Pseudo-Random Function (PRF) $\mathcal{F}_\kappa(v_0)$ operates by taking a secret key $\kappa$ and a public initial value $v_0$ as inputs, and produces a sequence of bits that appear to be uniformly random. One commonly used construction for a PRF is based on the AES counter mode (CTR mode). For practical purposes, the output is mapped to a vector over $\mathbb{Z}_p$. Thus, we denote the output vector of length $\ell$ with elements in $\mathbb{Z}_p$ as $\mathbf{y} = \mathcal{F}_\kappa(v_0, \ell)$. A PRF is considered secure if $\mathbf{y}$ is indistinguishable from a vector of length $\ell$ over $\mathbb{Z}_p$ that is uniformly random. Moreover, the output should be completely different when either $\kappa$ or $v_0$ is changed. A formal definition and security requirements can be found in Chapter 3.5 of \cite{katz2020introduction}.

KhPRF is a special type of PRF with additive key homomorphism:
\begin{equation}
    \mathcal{F}_{\kappa_1+\kappa_2}(v_0, \ell) = \mathcal{F}_{\kappa_1}(v_0, \ell) + \mathcal{F}_{\kappa_2}(v_0, \ell) \bmod p.
\end{equation}
Efficient KhPRFs are typically based on LWR or LWE constructions \cite{boneh2013key}. However, the KhPRF based on LWE is not exact; it actually satisfies an approximate additive key homomorphism:
\begin{equation}
    \mathcal{F}_{\kappa_1+\kappa_2}(v_0, \ell) = \mathcal{F}_{\kappa_1}(v_0, \ell) + \mathcal{F}_{\kappa_2}(v_0, \ell) + \mathbf{e} \bmod p.
\end{equation}
This is referred to as an $\epsilon$-almost KhPRF, where $\epsilon$ is the $L_1$-norm of the error vector $\mathbf{e}$.

Since in FL, $\ell$ typically needs to be large, an exact KhPRF cannot meet the efficiency requirements. Therefore, only an almost KhPRF can be used. As shown in the experimental results in \cite{karthikeyan2025opa}, the LWR-based scheme performs nearly two orders of magnitude faster than an exact KhPRF scheme based on the Hidden Subgroup Membership assumption. 
Hence, we adopt the almost KhPRF and propose a new encoding to handle the errors introduced by the almost KhPRF.

\subsection{Other Cryptographic Primitives}
A key-agreement (KA) protocol allows two parties to agree on a shared private key through public-key exchange. $P_i$ and $P_j$ each generate a public-private key pair, $(pk_i, sk_i) = \operatorname{KA.GenKey(\lambda)}$ and $(pk_j, sk_j) = \operatorname{KA.GenKey(\lambda)}$, where $\lambda$ is the security parameter, and then send their public keys, $pk_i$ and $pk_j$, to each other. $P_i$ uses $sk_i$ and $pk_j$ to compute the shared private key $sk_{i,j} = \operatorname{KA.Agree(sk_i, pk_j)}$. Similarly, $P_j$ uses $sk_j$ and $pk_i$ to compute $sk_{j,i} = \operatorname{KA.Agree(sk_j, pk_i)}$. The KA protocol requires that $sk_{i,j} = sk_{j,i}$. 
A secure KA ensures that the public keys $pk_i$ and $pk_j$ do not leak any information about the private keys $sk_i$ and $sk_j$. Additionally, no external party can obtain any information about $sk_{i,j}$.


A hash function accepts an input of arbitrary length and returns a fixed-length output. Let $y = \operatorname{Hash}(x)$. A secure hash function satisfies that given $y$, finding $x$ is computationally difficult, and finding $x' \neq x$ such that $y = \operatorname{Hash}(x')$ is also computationally hard. Furthermore, finding $x_1 \neq x_2$ such that $\operatorname{Hash}(x_1) = \operatorname{Hash}(x_2)$ is computationally difficult.

\begin{figure*}
    \centering
    \includegraphics[width=\linewidth]{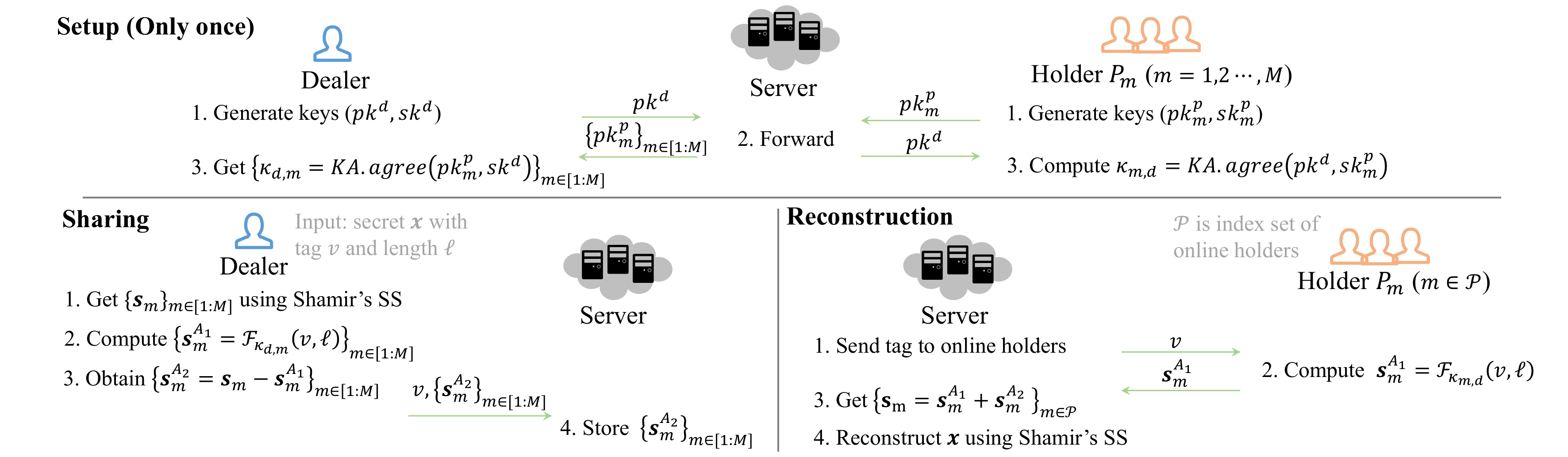}
    \caption{Details of our \TLSS. For simplicity, we assume that all $v_m$ equal to $v$. The dealer's secret is a vector $\mathbf{x}$ of length $\ell$. The setup phase only needs to be performed once, after which multiple secrets can be shared and reconstructed.}
    \label{fig:tlss_protocol}
\end{figure*}

\section{Two Layer Secret Sharing}
\label{sec:tlss}
The scenario we consider involves a dealer and the share holders $P_1, P_2, \dots, P_M$ who do not have direct communication channels between them and require a server for intermediary communication. After some time, the server needs to reconstruct the secret. This situation is quite common in cloud computing architectures that require collaboration among multiple users, especially in the context of FL \cite{bonawitz2017practical,bell2020secure,karthikeyan2025opa}. 
Our \TLSS scheme is specifically designed for such a secret sharing scenario. \TLSS consists of three phases: Setup, Sharing, and Reconstruction. We assume that communication between all participants is secure, for example, through SSL communication \cite{oppliger2023ssl}. We now describe the details of the three phases.

In the Setup phase, all participants share keys based on the security parameter $\lambda$. The dealer generates a public-private key pair $(pk^d, sk^d)$ and sends $pk^d$ to the server. Each holder $P_m$ generates a key pair $(pk^{p_m}, sk^{p_m})$ and sends $pk^{p_m}$ to the server.
The dealer then downloads all public keys $\{pk^{p_m}\}_{m \in [1:M]}$ from the server, and using $sk^d$, shares a key $\kappa_{d,m}$ with each holder $P_m$. Meanwhile, each holder downloads $pk^d$ and computes $\kappa_{m,d}$ using $sk^{p_m}$.

In the Sharing phase, the user shares the secret $x$ with all holders. The user first applies Shamir's SS to generate shares $s_m$ for each holder $P_m$. This involves sampling $T-1$ random numbers from $\mathbb{Z}_p$ to construct a secret polynomial $f(x)$, then computing $s_m = f(\beta_m) \bmod p$. The user then uses additive SS to share $s_m$ with $P_m$ and the server.
First, the user generates a random tag $v_m$ and uses the shared key $\kappa_{d,m}$ with $P_m$ to generate a pseudorandom number $s_m^{A_1} = \mathcal{F}_{\kappa_{d,m}}(v_m, 1)$. Then, the user computes $s_m^{A_2} = s_m - s_m^{A_1} \bmod p$. Finally, the user sends the tags $\{v_m\}_{m \in [1:M]}$ and the values $\{s_m^{A_2}\}_{m \in [1:M]}$ to the server. Here, $v_m$ serves merely as an identifier for different users' inputs and can be either the same for all holders or a public parameter agreed upon in advance. Thus, the only data actually sent in the Sharing phase are the values $\{s_m^{A_2}\}_{m \in [1:M]}$.

In the Reconstruction phase, the server obtains the shares from the share holders and reconstructs the secret. The server first sends the tag $v_m$ to the holder $P_m$. Then, the holder $P_m$ computes $s_m^{A_1,p} = \mathcal{F}_{\kappa_{m,d}}(v_m, 1)$ and sends it back to the server. The server then recovers the Shamir’s SS by computing $s_m' = s_m^{A_1,p} + s_m^{A_2} \bmod p$, and uses $(\beta_m, s_m')$ to reconstruct the secret $x$. Only if at least $T$ share holders successfully send their messages, can the server successfully reconstruct the true secret.

The correctness is obvious: KA ensures that the key $\kappa_{d,m}$ computed by the dealer is identical to the key $\kappa_{m,d}$ computed by the holder $P_m$, thereby guaranteeing that the derived value $s_m^{A_1,p}$ from the PRF by the holder $P_m$ is equivalent to the value $s_m^{A_1}$ computed by the dealer. Consequently, the server can accurately reconstruct the Shamir's secret share $s_m$.

\begin{equation*}
    \begin{aligned}
        s_m^\prime&=s_m^{A_1,p}+s_m^{A_2} \bmod p = \mathcal{F}_{\kappa_{m,d}}(v_m,1)+s_m^{A_2} \bmod p \\
        &=\mathcal{F}_{\kappa_{d,m}}(v_m,1)+s_m^{A_2} \bmod p = s_m^{A_1}+s_m^{A_2} \bmod p =s_m
    \end{aligned}
\end{equation*}
Therefore, the reconstruction algorithm of Shamir's SS guarantees that the server can accurately reconstruct the secret $x$. The detailed pseudocode and communication flow are provided in \cref{fig:tlss_protocol}.
Next, we analyze some useful and interesting properties of our \TLSS.

\subsection{Security}

To prove the security of \TLSS, we need to demonstrate two points. First, the shares of any two secrets are indistinguishable. Second, fewer than $T$ share holders colluding cannot obtain any information about the secret. The second point is evident, as \TLSS does not send any information to $\{P_m\}_{m \in [1:M]}$ during the Sharing phase. Therefore, when share holders collude, they can only access the keys negotiated during the Setup phase, and nothing else. Since the key and the secret $x$ are unrelated, even if all share holders collude, they cannot obtain any information about $x$.

For the first point, we have the following theorem:
\begin{theorem}
\label{thm:ss_security_distinguish}
    For any two secret inputs $x_1, x_2 \in \mathbb{Z}_p$, the distributions of the \TLSS shares $\{s_m^{A_2}\}_{m \in [1:M]}^{x_1}$ and $\{s_m^{A_2}\}_{m \in [1:M]}^{x_2}$ are computationally indistinguishable.
\end{theorem}

The proof of \cref{thm:ss_security_distinguish} is given in \cref{apx:proof_ss_security_distinguish}.

In addition to the dealer and share holders, our scheme also involves a server. Therefore, we consider the security of \TLSS when the server colludes with the share holders. We have the following theorem:
\begin{theorem}
\label{thm:ss_security_collusion}
    If the number of share holders compromised by the server is fewer than $T$, then the server cannot obtain any information about secret $x$ from its \TLSS shares $\{s_m^{A_2}\}_{m \in [1:M]}$.
\end{theorem}

\begin{proof}
    The proof here is straightforward. When the server colludes with $P_m$, it can obtain the value of $\kappa_{m,d}$, which allows the server to compute $s_m^{A_1}$. As a result, the server can compute $s_m = s_m^{A_1} + s_m^{A_2} \bmod p$. Therefore, when the server colludes with $t$ share holders, it can obtain $t$ Shamir's SS shares. Due to the security of KA, the keys of the non-colluding share holders will not be leaked, so the server cannot obtain any other Shamir shares. Based on the security of Shamir's SS, when $t < T$, the server cannot learn any information about $x$.
\end{proof}

\subsection{Additive Homomorphism}
Below, we prove that \TLSS satisfies the additive homomorphism property, which means that the shares of the sum of two secrets can be obtained from the shares of the individual secrets without the need for the dealer to re-execute the sharing algorithm. Intuitively, since both additive SS and Shamir's SS satisfy additive homomorphism, \TLSS inherits this property.

Let the dealer construct the secret polynomial $f_1(x)$ for secret $x_1$, and the secret polynomial $f_2(x)$ for secret $x_2$. Therefore, the Shamir's SS shares for the holder $P_m$ corresponding to secrets $x_1$ and $x_2$ are given by $
s_m^{x_1} = f_1(\beta_m) \bmod p$ and $s_m^{x_2} = f_2(\beta_m) \bmod p$.

Let the tags of the two secrets for $P_m$ be denoted as $v_m^{x_1}$ and $v_m^{x_2}$, respectively. Then, the \TLSS share values computed by $P_m$ are $s_m^{A_1,x_1} = \mathcal{F}_{\kappa_{m,d}}(v_m^{x_1}, 1)$ and $s_m^{A_1,x_2} = \mathcal{F}_{\kappa_{m,d}}(v_m^{x_2}, 1)$.
Since $\kappa_{m,d} = \kappa_{d,m}$, the shares received by the server are $\{s_m^{A_2,x_1} = s_m^{x_1} - s_m^{A_1,x_1} \bmod p\}_{m \in [1:M]}$ and $\{s_m^{A_2,x_2} = s_m^{x_2} - s_m^{A_1,x_2} \bmod p\}_{m \in [1:M]}$.

Next, we prove that the server and the share holders can locally compute the shares of $x_3 = x_1 + x_2 \bmod p$. A valid share of $x_3$ for each $P_m$ is given by:
\begin{equation*}
    s_m^{A_1,x_3} = s_m^{A_1,x_1} + s_m^{A_1,x_2} \bmod p
\end{equation*}
and the corresponding shares for the server are:
\begin{equation*}
    \{s_m^{A_2,x_3} = s_m^{A_2,x_1} + s_m^{A_2,x_2} \bmod p\}_{m \in [1:M]}.
\end{equation*}

We now show that after obtaining the share $s_m^{A_1,x_3}$ from $P_m$, the server can compute the corresponding Shamir's share of $x_3$, denoted as $f_3(\beta_m)$, where $f_3(x) = f_1(x) + f_2(x) \bmod p$.
\begin{equation*}
    \begin{aligned}
        s_m^{x_3} &= s_m^{A_1,x_3} + s_m^{A_2,x_3} \bmod p \\
        &= s_m^{A_1,x_1} + s_m^{A_1,x_2} + s_m^{A_2,x_1} + s_m^{A_2,x_2} \bmod p \\
        &= s_m^{x_1} + s_m^{x_2} \bmod p \\
        &= f_1(\beta_m) + f_2(\beta_m) \bmod p \\
        &= f_3(\beta_m) \bmod p.
    \end{aligned}
\end{equation*}
Thus, the server can reconstruct the secret $x_3$ using the points $\{(\beta_m, s_m^{x_3})\}_{m \in [1:M]}$.

\subsection{Batch Processing}
In \cite{braun2023secure, karthikeyan2025opa}, the concept of packed Shamir's SS was proposed. They demonstrated that multiple secrets can be hidden within a secret polynomial $f(x)$, while keeping the output size unchanged. The basic idea of packed Shamir's SS is to construct two polynomials: one is the secret polynomial $s(x)$, and the other is a random polynomial $r(x)$. The random polynomial is used to obscure the secret polynomial, resulting in a share polynomial, which is then distributed to the share holders. 

We can replace the Shamir’s SS in the construction of \TLSS with packed Shamir’s SS without affecting its properties, such as security and additive homomorphism. This is because the additive SS at the second layer is used for a one-time padding encryption, which ensures the security of shares. In the proof of additive homomorphism, we only rely on the additive homomorphism property of Shamir’s SS. If the packed Shamir's SS satisfies the additive homomorphism property at the first level, the constructed \TLSS will also satisfy this property.

In addition to using packed SS to improve the efficiency of \TLSS, we can further optimize the computation of shares by enhancing the PRF generation. Notice that each time we invoke the PRF, we need a tag to generate a pseudo-random number over $\mathbb{Z}_p$. When there are multiple Shamir’s SS share values, we can generate multiple pseudo-random numbers in one call, thus reducing the number of tags. Specifically, a secret vector can be processed using only one tag. We use $\mathcal{F}_{\kappa_{d,m}}(v_m,\ell)$ to generate a pseudo-random vector, allowing the second-level additive SS to handle multiple first-level share values at once. Note that a secure PRF requires unpredictability, meaning that the $\ell$ generated pseudo-random numbers are independent. Therefore, replacing with $\mathcal{F}_{\kappa_{d,m}}(v_m,\ell)$ does not affect the security of \TLSS.

\subsection{Theoretical Overhead}

In both \TLSS and Shamir's SS, the number of secrets is proportional to the overhead. Therefore, we analyze the primary computational and communication overheads involved in sharing a secret $x \in \mathbb{Z}_p$. Since a tag is only required to uniquely identify each secret, our analysis below does not consider the overhead of tags. 

We provide a comparison of the overall overheads for sharing and reconstruction in \cref{tab:tlss_theory_compare}. As shown, our scheme significantly reduces the communication overhead during the sharing process while maintaining comparable overheads for other operations. The detailed analysis is provided below.

\subsubsection{Computational Overhead}

In the Setup phase, the dealer needs to perform $M$ key exchange operations with each share holder using the KA protocol to share a private key. The dealer must compute a key and engage in $M$ rounds of KA. The computational overhead for a secret share holder $P_m$ involves generating a public key and performing a single key exchange. The server only forwards the public keys and does not require additional computation. It is important to note that this phase is identical in both \TLSS and Shamir's SS. The difference lies in that, in Shamir's SS, the shared key is used for AE encryption, whereas in \TLSS, it is used to generate pseudo-random numbers.

In the Sharing phase, in \TLSS, the dealer needs to invoke a polynomial evaluation once, followed by $M$ calls of the PRF function, and then perform $M$ modular additions. The server only needs to store the $M$ received share values (elements in $\mathbb{Z}_p$). Each share holder $P_m$ does not require any computation. In Shamir’s SS, the dealer must perform a polynomial evaluation once, followed by $M$ AE encryptions. The server must forward all the AE ciphertexts to the corresponding share holders, with each share holder $P_m$ storing the AE ciphertext.

In the Reconstruction phase, in \TLSS, each online share holder $P_m$ calls the PRF once. The server then performs $M$ modular additions and performs a polynomial interpolation once. In Shamir’s SS, each online share holder $P_m$ decrypts the AE ciphertext once. The server needs a polynomial interpolation.

The overhead for modular addition is negligible compared to other operations, and the overheads of the PRF (implemented by AES) and AE are comparable. Therefore, the computational overhead of \TLSS and Shamir's SS is similar.

\subsubsection{Communication Overhead}

In the Setup phase, the dealer needs to upload a public key and download $M$ public keys. Each share holder $P_m$ incurs communication overhead by uploading a public key and downloading a public key. The server needs to receive $M+1$ public keys and forward $M+1$ public keys. Since the Setup phase in \TLSS and Shamir's SS is identical, the communication overhead is also the same.

In the Sharing phase, in \TLSS, the dealer sends $M$ share values. In Shamir's SS, the dealer sends $M$ AE ciphertexts. The server receives these $M$ ciphertexts and then forwards them to each share holder, thus requiring $M$ additional AE ciphertexts to be sent. Each share holder $P_m$ needs to receive one AE ciphertext.
In the Reconstruction phase, both \TLSS and Shamir’s SS require only one share value to be sent by each online share holder $P_m$.

Although AE ciphertexts are slightly larger than elements in $\mathbb{Z}_p$, the difference is minimal. However, it is noteworthy that in \TLSS, during the Sharing phase, the server does not need to distribute data, and share holders do not need to receive data. In fact, share holders may even be offline. Therefore, the communication advantage of \TLSS lies in reducing the server's communication by nearly half and eliminating the need for share holders to download data.

\begin{table}[!t]
\caption{Overall computational (comp.) and communication (comm.) overhead comparison. Eval is the polynomial evaluation overhead, Intep. is the interpolation overhead.}
\label{tab:tlss_theory_compare}
\begin{tabular}{@{}cccc@{}}
\toprule
 & Scheme & Sharing & Reconstruction \\ \midrule
\multirow{2}{*}{Comp.} &  Shamir & Eval,$M$ AE-Enc. & $M$ AE-Dec., Interp. \\
 & \TLSS & Eval, $M$ PRF, $M$ Add & $M$ PRF, $M$ Add, Interp. \\ \cmidrule(){1-4} 
\multirow{2}{*}{Comm.} &  Shamir & $2M$ AE-Ctxt & $M$ Field-Elements \\
 & \TLSS & $M$ Field-Elements & $M$ Field-Elements \\ \bottomrule
\end{tabular}
\end{table}

\section{Non-Forward Secure Aggregation}

At a high level, our protocol primarily uses KhPRF to achieve one-shot secure aggregation. First, each user locally generates a KhPRF key, and then uses the pseudo-random mask generated by KhPRF to obscure the user's input. The user then shares the key using \TLSS and sends the masked input to the server. The decryptor, based on the additive homomorphism of \TLSS, computes the sum of all online users' keys. The server then reconstructs the sum of the keys and generates a global mask based on KhPRF. Finally, the server performs an unmasking operation on the aggregated result to obtain the final outcome.

In the following subsections, we first describe the system model and threat model, followed by the presentation of our novel encoding for almost KhPRF. We then provide the details of the secure aggregation protocol and the analysis of its security.

\subsection{System Model and Threat Model}

Our assumptions and goals align with those of most current single-server protocols, such as \cite{ma2023flamingo, karthikeyan2025opa}.
In our secure aggregation protocol, there are three types of entities:

\begin{itemize}
    \item \textbf{Users:} There are $N$ users, $U_1, U_2, \dots, U_N$. Each user owns a private dataset and has the capability to train a model, but the network connection may not be stable. Users can go offline at any time, and their network bandwidth may be limited.
    \item \textbf{Server:} There is a single server, usually provided by a third-party service provider. The server offers stable service and does not go offline. The server aggregates models from different users.
    \item \textbf{Decryptors:} There are $M$ decryptors,  $P_1, P_2, \dots, P_M$. These decryptors can be selected from the users or be separate auxiliary nodes. Decryptors may also go offline, and their network bandwidth is limited. They assist the server in reconstructing the globally aggregated model.
\end{itemize}

\textbf{Communication Assumptions.} We assume that there is no direct communication channel between different users. Each user or decryptor has a secure communication channel with the server, such as SSL \cite{oppliger2023ssl}. This means that the transmitted content is only known to the sender and the receiver, and cannot be intercepted or altered by other parties. Therefore, we do not consider attackers from outside the system.

\textbf{Threat Model Assumptions.} We assume that all entities are semi-honest, meaning they follow the protocol steps correctly, but may perform additional computations to infer other entities' private information from the received data. Additionally, we assume that at most $T-1$ decryptors and the server may collude. 
In our scheme, the dropout and corruption thresholds for decryptors differ from those in OPA. In the maliciously-secure OPA, robustness against Byzantine faults requires that $D + C < \frac{M}{3}$, where $D$ represents the number of dropout decryptors and $C$ denotes the number of corrupted decryptors. However, we operate within a semi-honest model, where decryptors adhere to the protocol and provide valid shares. As a result, reconstruction only requires $T$ responses. In this case, up to $D = M - T$ dropouts are tolerated, with a condition that $C < T$ to ensure privacy.

\textbf{Design Goals.} Our protocol has three main goals.
\textbf{(1) Correctness:} The first goal is to ensure the correctness of the aggregation, meaning that the aggregation function in \cref{eq:fl_aggregation} is computed correctly.
\textbf{(2) Security:} The second goal is to ensure that the aggregation is secure, meaning that under our threat model, the private input of users will not be exposed to others, except for the aggregation result and the inputs of compromised users.
\textbf{(3) Efficiency:} The third goal is to minimize both the computational and communication overheads for the users, decryptors and the server.

\subsection{CRT Encoding for Almost KhPRF}

\subsubsection{Error Analysis of Almost KhPRF}
Since the output of KhPRF satisfies approximate key-homomorphism, although the error $\mathbf{e}$ from a single addition is relatively small, repeated additions accumulate errors. The error $\mathbf{e}$ increases linearly with the number of additions. In previous works, such as \cite{li2023lerna, guan2025opsa, karthikeyan2025opa}, researchers often mitigate the error introduced by addition by scaling the inputs with an amplification factor $\Delta$. This approach suppresses the error during decoding.

For the reader's convenience, we first present the simple construction of the KhPRF based on LWR as proposed by \cite{boneh2013key} in the random oracle model. Let $q$, $p_r$, and $n_\lambda$ be integers.
Let $\lfloor x \rfloor_{p_r}=\lfloor x \cdot \frac{p_r}{q} \rfloor \bmod p_r$ be the rounding function. Let $H_2(x): \chi \rightarrow \mathbb{Z}_q^{n_\lambda}$ be a hash function from input to $\mathbb{Z}_q^{n_\lambda}$. Then the KhPRF is defined as follows:
\begin{equation}
\label{eq:LWR_random_oracle}
    F_{H}(\mathbf{k},x)=\lfloor \langle H_2(x),\mathbf{k}\rangle \rfloor_{p_r},
\end{equation}
where $\langle \cdot,\cdot \rangle$ is the inner product of two vectors with modulo $q$. As shown in \cite{boneh2013key}, this construction satisfies the following property:
\begin{equation}
    F_{H}(\mathbf{k}_1+\mathbf{k}_2,x)=F_{H}(\mathbf{k}_1,x)+F_{H}(\mathbf{k}_2,x)+e,
\end{equation}
where $e \in \{0,1\}$. It is easy to prove. Let $y_1 = \langle H_2(x), \mathbf{k}_1 \rangle$ and $y_2 = \langle H_2(x), \mathbf{k}_2 \rangle$. 
Due to the linearity of the inner product, we have $\langle H_2(x),\mathbf{k}_1+\mathbf{k}_2 \rangle=y_1+y_2 \bmod q$.
Let $y_1p_r = I_1q + r_1$ and $y_2p_r = I_2q + r_2$, where $I_1, I_2 \in \mathbb{Z}$ and $r_1, r_2 \in \mathbb{Z}_q$. Therefore, we have $F_H(\mathbf{k}_1, x) = \lfloor y_1 \rfloor_{p_r} = I_1$ and $F_H(\mathbf{k}_2, x) = \lfloor y_2 \rfloor_{p_r} = I_2$.
Then we have the following:
\begin{equation*}
    \begin{aligned}
        F_H(\mathbf{k}_1+\mathbf{k}_2,x)&=\lfloor y_1+y_2 \rfloor_{p_r}\\
        &=\lfloor I_1+\frac{r_1}{q}+I_2+\frac{r_2}{q} \rfloor\\
        &=I_1+I_2+\lfloor \frac{r_1+r_2}{q} \rfloor \\
        &=F_H(\mathbf{k}_1,x)+F_H(\mathbf{k}_2,x)+e
    \end{aligned}
\end{equation*}
where $e=\lfloor \frac{r_1+r_2}{q} \rfloor$. Since $r_1 + r_2 \in [0, 2q - 2]$, it follows that $e \in \{0, 1\}$.

Next, we analyze the situation when multiple additions are performed. Suppose there are
$n$ keys $\mathbf{k}_1,\cdots,\mathbf{k}_2,\cdots,\mathbf{k}_n$. Let $\mathbf{k}=\sum_{i=1}^n\mathbf{k}_i$. Let $y_i=\langle H_2(x),\mathbf{k}_i \rangle$ and $y_ip_r=I_iq+r_i$ where $I_i \in \mathbb{Z}$ and $r_i \in \mathbb{Z}_q$.  Then we have
\begin{equation*}
    \begin{aligned}
        F_H(\mathbf{k},x)&=\lfloor \sum_{i=1}^n y_i  \rfloor_{p_r} = \lfloor  \sum_{i=1}^n I_i +\frac{\sum_{i=1}^n r_i}{q} \rfloor\\
        &=\sum_{i=1}^n I_i + \lfloor \frac{\sum_{i=1}^n r_i}{q} \rfloor = \sum_{i=1}^n F_H(\mathbf{k}_i,x)+e
    \end{aligned}
\end{equation*}
where $e=\lfloor \frac{\sum_{i=1}^n r_i}{q} \rfloor$.

Since $r_i \in [0, q-1]$, it follows that $\sum_{i=1}^n r_i \in [0, nq - n]$. Therefore, $e \leq \lfloor \frac{nq - n}{q} \rfloor = n - \lfloor \frac{n}{q} \rfloor$. Typically, $q$ is much larger than $n$, so $e \in [0, n]$. 
Note that since different keys are independent, the values of different $r_i$ are also independent, and their values do not influence each other. Thus, it is possible for $\sum_{i=1}^n r_i = nq - n$. In other words, $\vert e \vert \leq n$ is a supremum (least upper bound).

In other KhPRF constructions based on LWR or RLWR, such as the KhPRF in the standard model presented in \cite{boneh2013key}, or the RLWE KhPRF variant in \cite{banerjee2012pseudorandom}, the rounding function $\lfloor \cdot \rfloor_{p_r}$ is the sole operation that introduces error.
Therefore, the error analysis in these constructions applies to the analysis presented above. We thus have the following theorem:

\begin{theorem}
    For any almost KhPRF $F_H(\mathbf{k},x)$ based on rounding function $\lfloor \cdot \rfloor_{p_r}$, we have 
    \begin{equation*}
        F_H(\sum_{i=1}^n \mathbf{k}_i,x)=\sum_{i=1}^n F_H(\mathbf{k}_i,x)+e
    \end{equation*}
    where $e \in [0,n]$.
\end{theorem}

Let the private input be denoted as $x_i$, and $\tau$ be the public parameter.
Existing methods \cite{li2023lerna,guan2025opsa,karthikeyan2025opa} mask the input $x_i$ with $F_H(\mathbf{k}_i,\tau)$ as follows:
\begin{equation}
\label{eq:existing_mask}
    y_i=\Delta x_i- F_H(\mathbf{k}_i,\tau) \bmod p_r.
\end{equation}

If there are $n$ inputs involved in the aggregation, then the final aggregation result is:
\begin{equation}
    \sum_{i=1}^nx_i = \lfloor \frac{1}{\Delta} \times (\sum_{i=1}^ny_i+  F_H( \sum_{i=1}^n\mathbf{k}_i,\tau) \bmod p_r) \rfloor.
\end{equation}

To correctly decode, $\Delta$ must satisfy $\Delta > n$. This means that each input $x_i$ will increase by at least $\log_2(n+1)$. When the input is a $d$-dimensional vector $\mathbf{x}_i$, the increase will be $d\log_2(n+1)$. Therefore, the masking method used in existing approaches, as shown in \cref{eq:existing_mask}, will significantly increase the communication overhead for users  when the number of users involved in the aggregation is large.

\subsubsection{CRT Encoding}
Our solution is to combine several elements together to reduce the input dimensionality, thereby lowering the communication overhead. Clearly, the merging method for these elements must satisfy additively homomorphic properties, otherwise, it will not support the aggregation of inputs from different users. We propose using the Chinese Remainder Theorem to merge elements of different dimensions.

Let $p_1, p_2, \dots, p_{d_c} > p_m$ be $d_c$ pairwise coprime positive integers, where $\mathbb{Z}_{p_m}$ is the plaintext space. Let $p_c = \prod_{i=1}^{d_c} p_i$. For the encoding of an input $\mathbf{x} \in \mathbb{Z}_{p_m}^{d_c}$, we first represent each element $\mathbf{x}[i]$ extended to $\mathbb{Z}_{p_i}$, and then use the Chinese Remainder Theorem to combine the $d_c$ elements into a single element in $\mathbb{Z}_{p_c}$. Below, we provide a detailed description of the packing process.

\textbf{Representation Extending.}
As described in \cref{apx:pre_sa}, we use the interval $(0,\frac{p_m}{2})$ to represent positive integers and $(\frac{p_m}{2},p_m)$ to represent negative integers. The actual negative numbers lie in $(-\frac{p_m}{2},0)$. In practical applications, we do not require the results of modular operations over $\mathbb{Z}_{p_m}$, but rather the true additive results. Therefore, when configuring the parameters, the maximum value of the final aggregated result is typically less than $\frac{p_m}{2}$, and the minimum value is greater than $-\frac{p_m}{2}$. Consequently, the effective range is $(-\frac{p_m}{2},\frac{p_m}{2})$. For non-negative numbers, we leave them unchanged, whereas for negative numbers, we map them to $(\frac{p_i}{2},p_i)$ by adding $p_i$. 

We denote this process of representation extension by $\phi(x,p_i)$. During decoding, if a result is less than $\frac{p_i}{2}$, it remains unchanged; otherwise, we subtract $p_i$. We denote this inverse mapping by $\phi^{-1}(x,p_i)$. The following lemma holds:

\begin{lemma}
\label{lem:promote}
    If $x_1, x_2 \in (-\frac{p_m}{2},\frac{p_m}{2})$ and $x_1 + x_2 \in (-\frac{p_m}{2},\frac{p_m}{2})$, then the following equality holds:
    \begin{equation*}
        x_1+x_2=\phi^{-1}(\phi(x_1,p_i)+\phi(x_2,p_i)\bmod p_i,p_i),
    \end{equation*}
    where $p_i>p_m$.
\end{lemma}

\begin{proof}
    When $x_1, x_2 \in [0, \frac{p_m}{2})$ and $x_1 + x_2 \in [0, \frac{p_m}{2})$, since $p_i > p_m$, the operations over $\mathbb{Z}_{p_i}$ are identical to those before extension.  

    When $x_1 \in [0, \frac{p_m}{2})$ and $x_2 \in (-\frac{p_m}{2}, 0)$, after extension, $x_1$ remains unchanged, while $\phi(x_2, p_i) = x_2 + p_i \in (\frac{p_i}{2}, p_i)$. If $x_1 + x_2 \geq 0$, then $x_1 + x_2 + p_i > p_i$, so the result over $\mathbb{Z}_{p_i}$ is 
    \[
        x_1 + \phi(x_2, p_i) \bmod p_i = x_1 + x_2<\frac{p_i}{2},
    \] which is mapped back to $x_1+x_2$.
    If $x_1 + x_2 < 0$, then $x_1 + x_2 + p_i < p_i$, so the computation over $\mathbb{Z}_{p_i}$  yields $x_1 + x_2 + p_i$. Because $x_1\geq0$, $x_1+x_2+p_i>p_i-\frac{p_m}{2}>\frac{p_i}{2}$. Hence, applying the inverse mapping gives 
    \[
        \phi^{-1}(x_1 + x_2 + p_i,p_i) = x_1 + x_2.
    \]  

    When $x_1, x_2 \in (-\frac{p_m}{2}, 0)$ and $x_1 + x_2 \in (-\frac{p_m}{2}, 0)$, we have $\phi(x_1) = x_1 + p_i$ and $\phi(x_2) = x_2 + p_i$. Since $p_i > \frac{p_m}{2}$, we have $x_1 + x_2 + 2 p_i > p_i$ and $x_1 + x_2 + 2 p_i < 2 p_i$. Therefore, the result over $\mathbb{Z}_{p_i}$ is $x_1 + x_2 + p_i>\frac{p_i}{2}$, and applying the inverse mapping yields 
    \[
        \phi^{-1}(\phi(x_1) + \phi(x_2) \bmod p_i,p_i) = x_1 + x_2.
    \]  
     This completes the proof of \cref{lem:promote}.
\end{proof}

\textbf{CRT Packing.}
Let the input be a vector $\mathbf{x} \in \mathbb{Z}_{p_m}^{d_c}$ of length $d_c$. After extending each entry $\mathbf{x}[i]$ to $\mathbb{Z}_{p_i}$, we can construct the following system of congruences:
\begin{equation}
    \left \{
    \begin{aligned}
        x& \equiv \mathbf{x}[1]  \bmod p_1\\
        x& \equiv \mathbf{x}[2]  \bmod p_2\\
        &\cdots\\
        x&\equiv \mathbf{x}[d_c]  \bmod p_{d_c}\\
    \end{aligned}
    \right.
\end{equation}
Let $M_i = \frac{p_c}{p_i}$ (since $p_c = \prod_{i=1}^{d_c} p_i$, $M_i$ is an integer).  
Because any $p_j$ and $p_i$ are co-prime for $i \neq j$, there exists an element $y_i \in \mathbb{Z}_{p_i}$ such that $M_i \times y_i \equiv 1 \bmod{p_i}$.
This element $y_i$ can be efficiently computed using the extended Euclidean algorithm. Hence, we can solve:
\begin{equation}
\label{eq:crt_soluation}
    x\equiv \sum_{i=1}^{d_c} \mathbf{x}[i]\times M_i\times y_i \bmod p_c.
\end{equation}
Thus, we have packed the $d_c$ elements into a single element in $\mathbb{Z}_{p_c}$.  
Unpacking is straightforward: taking the value modulo $p_i$ directly recovers the element stored at the $i$-th position.

We denote the packing process by $\operatorname{CRT}(\mathbf{x})$ and the unpacking process by $\operatorname{CRT}^{-1}(x)$. The following theorem holds:
\begin{theorem}
    Let $\mathbf{x}_1, \mathbf{x}_2 \in \mathbb{Z}_{p_m}^{d_c}$. If each $\mathbf{x}_1[i], \mathbf{x}_2[i]$ satisfies the conditions in \cref{lem:promote}, then the following equality holds:
    \begin{equation*}
        \mathbf{x}_1+\mathbf{x_2} = \operatorname{CRT}^{-1}(\operatorname{CRT}(\mathbf{x}_1)+\operatorname{CRT}(\mathbf{x}_2)\bmod p_c).
    \end{equation*}
\end{theorem}
The proof here is straightforward.  
$M_i$ and $y_i$ are constants, so \cref{eq:crt_soluation} is additively homomorphic. Moreover, by \cref{lem:promote}, the transformation from $\mathbb{Z}_{p_m}$ to $\mathbb{Z}_{p_i}$ is also additively homomorphic. By the composability of additive homomorphism, our encoding is therefore additively homomorphic as well.  
In other words, such packing does not affect the correctness of the final aggregation result. We can then mask the input $x$ using the KhPRF, following \cref{eq:existing_mask}.

We now analyze the communication cost. The original input $\mathbf{x} \in \mathbb{Z}_{p_m}^{d_c}$ requires $d_c \times \lceil \log_2 p_m \rceil$ bits for representation. After masking using \cref{eq:existing_mask}, the input becomes an element over $\mathbb{Z}_{\Delta p_m}^{d_c}$, requiring $d_c \times \lceil \log_2 \Delta p_m \rceil$ bits.  

After CRT packing, only $\lceil \log_2 \Delta p_c \rceil$ bits are needed. Although the co-primality condition must be satisfied, when $d_c$ is not very large, we have $p_i \approx p_m$. Therefore,
\[
    \log_2 p_c = \sum_{i=1}^{d_c} \log_2 p_i \approx d_c \log_2 p_m.
\]  
Hence, the theoretical bit savings achieved is:
\begin{equation*}
    d_c\times \lceil \log_2 \Delta p_m \rceil -\lceil \log_2 \Delta p_c \rceil \approx (d_c-1)\log_2\Delta.
\end{equation*}

For ease of implementation and efficient computation, it is common to choose a relatively small $d_c$. When the input dimension $d$ exceeds $d_c$, we partition it into $\lceil \frac{d}{d_c} \rceil$ groups, each packed into a single element over $\mathbb{Z}_{p_c}$.  
Consequently, the theoretical bit savings can be approximated as:
\begin{equation*}
      \lceil \frac{d}{d_c} \rceil\times (d_c-1)\log_2\Delta \approx (d-\lceil \frac{d}{d_c} \rceil) \times \log_2 \Delta.
\end{equation*}

\textbf{Parameter Selection.}
In KhPRF, the parameters that need to be determined are $p_m$, $\Delta$, $q$, and the dimension of the polynomial
ring dimension $n_\lambda$ if using KhPRF based on RLWR as described in \cref{apx:KhPRF_rlwr}. Plaintext modulus $p_m$ represents the input range, and $\Delta$ depends on $n$ (the maximum number of additions required for aggregation). The minimum value of $\Delta$ is $\lceil n+1 \rceil$ to ensure correctness. The polynomial ring dimension and $q$ are determined by the required security level, using the LWE estimator; further details can be found in \cref{apx:parameter_select_KhPRF}. When CRT packing is needed, we first fix $n_\lambda$ and the maximum allowable $q$ based on security requirements, and then determine the largest feasible $d_c$ given $\Delta$ and $p_m$. Subsequently, we select $d_c$ consecutive primes larger than $p_m$ as the CRT moduli.

\subsection{Details of \NFSA}

\begin{algorithm}[!tb]
\caption{Offline Phase of \NFSA}
\label{alg:offline}
\begin{algorithmic}[1]

\Statex \textbf{Input:} Security parameter $\lambda$
\State \textbf{For each} user $U_i$ \textbf{do:}
\State \quad Generate key pair $(pk_i^u, sk_i^u)=\operatorname{KA.GenKey(\lambda)}$
\State \quad Send $pk_i^u$ to the server
\State \textbf{EndFor}
\State \textbf{For each} decryptor $P_m$ \textbf{do:}
\State \quad Generate key pair $(pk_m^p, sk_m^p)=\operatorname{KA.GenKey(\lambda)}$
\State \quad Send $pk_m^p$ to the server
\State \textbf{EndFor}
\State \textbf{For each} user $U_i$ \textbf{do:}
\State \quad Download $\{ pk_m^p \}_{m \in [1:M]}$ from the server
\State \quad Compute $\{ \kappa_{i,m} = \operatorname{KA.Agree}(sk_i^u, pk_m^p) \}_{m \in [1:M]}$
\State \textbf{EndFor}
\State \textbf{For each} decryptor $P_m$ \textbf{do:}
\State \quad Download $\{ pk_i^u \}_{i \in [1:N]}$ from the server
\State \quad Compute $\{ \kappa_{m,i} = \operatorname{KA.Agree}(sk_m^p, pk_i^u) \}_{i \in [1:N]}$
\State \textbf{EndFor}
\end{algorithmic}
\end{algorithm}

\begin{algorithm}[!t]
\caption{Online Phase of \NFSA(The $r$-th Round)}
\label{alg:online}
\begin{algorithmic}[1]
\Statex \textbf{Input:} Each user $U_i$ has input $\mathbf{x}_i^{(r)} \in \mathbb{Z}_{p_m}^d$
\State \textbf{For each} user $U_i$ \textbf{where} $i \in [1:N]$ \textbf{do:}
\State \quad Generate a KhPRF key $\mathbf{k}_i^{(r)} \in \mathbb{Z}_q^{n_\lambda}$
\State \quad Derive temporary key $\kappa_{i,m}^{(r)} = \operatorname{Hash}(\kappa_{i,m} \Vert r)$
\State \quad Share $\mathbf{k}_i^{(r)}$ using \TLSS, taking $\kappa_{i,m}^{(r)}$ as the key of PRF, which returns the shares $\{\mathbf{s}_m^{A_2,i,r}\}_{m \in [1:M]}$
\State \quad Get $\mathbf{c}_i^{(r)} = \Delta \times \operatorname{CRT}(\mathbf{x}_i^{(r)}) - F_H(\mathbf{k}_i^{(r)}, r) \bmod \Delta p_c$
\State \quad Send $\{\mathbf{s}_m^{A_2,i,r}\}_{m \in [1:M]}$ and $\mathbf{c}_i^{(r)}$ to the server
\State \textbf{EndFor}

\State \textbf{Server:}
\State \quad Receive $\{\{\mathbf{s}_m^{A_2,i,r}\}_{m \in [1:M]}\}_{i\in \mathcal{U}}$ from online users $\mathcal{U}$
\State \quad Send $\mathcal{U}$ to the online decryptors $\{P_m\}_{m \in \mathcal{P}_1}$

\State \textbf{For each decryptor} $P_m, m \in \mathcal{P}_1$ \textbf{do:}
\State \quad Derive keys $\{\kappa_{m,i}^{(r)} = \operatorname{Hash}(\kappa_{m,i} \Vert r)\}_{i \in \mathcal{U}}$
\State \quad Compute $\mathbf{s}_m^{A_1,r} = \sum_{i \in \mathcal{U}} \mathcal{F}_{\kappa_{m,i}^{(r)}}(r, n_\lambda) \bmod q$
\State \quad Send $\mathbf{s}_m^{A_1,r}$ to the server
\State \textbf{EndFor}

\State \textbf{Server:}
\State \quad Receive $\{\mathbf{s}_m^{A_1,r}\}_{m \in \mathcal{P}_2}$ from online decryptors $\mathcal{P}_2$
\State \quad Compute $\mathbf{s}_m^{A_2,r} = \sum_{i \in \mathcal{U}} \mathbf{s}_m^{A_2,i,r} \bmod q$ for $m \in \mathcal{P}_2$
\State \quad Reconstruct KhPRF key $\mathbf{k}^{(r)}$ using \TLSS
\State \quad Compute $\mathbf{c}^{(r)} = \sum_{i \in \mathcal{U}} \mathbf{c}_i^{(r)} \bmod \Delta p_c$
\State \quad Compute $\mathbf{y}^{(r)}=F_H(\mathbf{k}^{(r)}, r) + \mathbf{c}^{(r)} \bmod \Delta p_c $
\State \quad Get the result $\mathbf{x}^{(r)} = \operatorname{CRT}^{-1} \left( \left\lfloor \frac{1}{\Delta} \times \mathbf{y}^{(r)} \right\rfloor \right)$
\end{algorithmic}
\end{algorithm}

Our secure aggregation protocol consists of two phases: the offline phase and the online phase. The offline phase corresponds to the Setup phase of \TLSS, where each user $U_i$ and each decryptor $P_m$ engage in key exchange to establish a shared key $\kappa_{i,m}$. Since this process is identical to the setup phase described in \cref{sec:tlss}, we omit further details, which can be found in \cref{alg:offline}. Our primary focus is on the online phase.

Before each round of secure aggregation begins, each participating user downloads the current global model from the server and trains it using their local dataset to obtain the local model parameters. These parameters are then encoded as elements in $\mathbb{Z}_{p_m}^d$, as described in \cref{apx:pre_sa}. Let the current round be denoted as $r$, and the input of user $U_i$ be represented by the integer vector $\mathbf{x}_i^{(r)}$. The secure aggregation process consists of two steps: Masking and Unmasking.

\textbf{Masking.} 
In this step, each user $U_i$ masks their input to obtain the ciphertext, which is then uploaded to the server. User $U_i$ randomly samples a KhPRF key $\mathbf{k}_i^{(r)} \in \mathbb{Z}_q^{n_\lambda}$ and shares it using \TLSS over $\mathbb{Z}_q$. Here, the round index $r$ is used as the tag. The value of the share for $P_m$ is computed as
\begin{equation*}
    \mathbf{s}_m^{A_1,i,r} = \mathcal{F}_{\kappa_{i,m}^{(r)}}(r,n_\lambda),
\end{equation*}
where $\kappa_{i,m}^{(r)} = \operatorname{Hash}(\kappa_{i,m} \Vert r)$. The round-specific shared key is derived using the hash function to avoid invoking the \TLSS setup phase for each aggregation round. Then, user $U_i$ computes the share values $\{\mathbf{s}_m^{A_2,i,r}\}_{m \in [1:M]}$ for the server.

The input is first encoded using our CRT encoding and masked using the KhPRF function: 
\begin{equation*}
    \mathbf{c}_i^{(r)} = \Delta \times \operatorname{CRT}(\mathbf{x}_i^{(r)}) - F_H(\mathbf{k}_i^{(r)},r) \mod \Delta p_c.
\end{equation*}
Both the share values $\{\mathbf{s}_m^{A_2,i,r}\}_{m \in [1:M]}$ and the ciphertext $\mathbf{c}_i^{(r)}$ are uploaded to the server. In this step, each user only needs to upload their data once.

\textbf{Unmasking.} 
In this step, the decryptors upload the auxiliary information to the server, which then recovers the global aggregation result. First, the server sends the list of online users, $\mathcal{U}$, to all online decryptors $\mathcal{P}_1$. For each decryptor $P_m$ where $m \in \mathcal{P}_1$, the shared key between the current round and user $U_i$ where $i \in \mathcal{U}$ is derived as $\kappa_{m,i}^{(r)} = \operatorname{Hash}(\kappa_{m,i} \Vert r)$. The decryptor $P_m$ then computes the share $\mathbf{s}_m^{A_1,i,r} = \mathcal{F}_{\kappa_{m,i}^{(r)}}(r, n_\lambda)$ for $U_i$.

Note that instead of directly returning the share values, each decryptor $P_m$ utilizes the additive homomorphism of \TLSS to compute the share of $\mathbf{k}^{(r)} = \sum_{i \in \mathcal{U}} \mathbf{k}_i^{(r)}$ as follows:
\begin{equation*}
    \mathbf{s}_m^{A_1,r} = \sum_{i \in \mathcal{U}} \mathbf{s}_m^{A_1,i,r} \bmod q.
\end{equation*}
The decryptor $P_m$ then sends $\mathbf{s}_m^{A_1,r}$ to the server. The server aggregates the shares uploaded by the users to obtain
\begin{equation*}
    \mathbf{s}_m^{A_2,r} = \sum_{i \in \mathcal{U}} \mathbf{s}_m^{A_2,i,r} \bmod q,
\end{equation*}
where $m \in \mathcal{P}_2$. Here, $\mathcal{P}_2$ denotes the indices of the decryptors who successfully uploaded their shares.

Finally, using the reconstruction algorithm of \TLSS, the server reconstructs the key $\mathbf{k}^{(r)}$. The server then calls the KhPRF function to obtain the global mask, $F_H(\mathbf{k}^{(r)}, r)$. The ciphertexts uploaded by the users are aggregated and unmasked as follows:
\begin{equation*}
    \mathbf{y}^{(r)} = F_H(\mathbf{k}^{(r)}, r)+\sum_{i \in \mathcal{U}} \mathbf{c}_i^{(r)} \bmod \Delta p_c.
\end{equation*}
Finally, the global aggregation result is decoded as
\begin{equation*}
    \mathbf{x}^{(r)} = \operatorname{CRT}^{-1}\left(\left\lfloor \frac{1}{\Delta} \times \mathbf{y}^{(r)} \right\rfloor \right).
\end{equation*}
The detailed process of the online phase is provided in \cref{alg:online}.

\subsection{Security Analysis of \NFSA}
As with most current one-server secure aggregation schemes \cite{bonawitz2017practical,bell2020secure,karthikeyan2025opa}, our security objective is to protect the privacy of user inputs, but not the aggregated results of online users. In other words, the server knows the aggregation result for each round. While attacks such as exploiting user choices could potentially reveal some privacy information about the users \cite{deer2022multi}, there are already protection schemes for this type of attack, such as those in \cite{deer2022multi,so2023securing}, which are compatible with our approach. The user selection strategy in these schemes can be directly applied to our protocol. This is beyond the scope of this work, so we do not discuss such attacks.

What we need to prove is that the input $\{\mathbf{x}_i^{(r)}\}_{r \in [1:R]}$ of any user $U_i$ does not leak to others, except for the aggregated result of the online users in each round $\{\mathbf{x}^{(r)}\}_{r \in [1:R]}$. User $U_i$ does not send any messages about their input $\{\mathbf{x}_i^{(r)}\}_{r \in [1:R]}$ to other users $U_j$ or decryptors $P_m$, and the public key $pk_i^u$ is independent of the private input, ensuring that the privacy of $U_i$ is not disclosed to other users or decryptors. Therefore, the only potential attacker is the server.

Intuitively, the security of \NFSA is based on the security guarantees provided by the KA protocol, the hash function, \TLSS, and the KhPRF function. The security of the KA protocol ensures that the shared key between the user and decryptor remains confidential. The security of the hash function guarantees that the temporary keys derived for each round are secure. The security of \TLSS ensures that the KhPRF key is not leaked. Finally, the security of the KhPRF function ensures that the masked ciphertext $\mathbf{c}_i$ does not reveal the user's privacy. Formally, we have the following theorem:

\begin{theorem}
\label{thm:security_nfsa}
    Assume the existence of a secure KA protocol, hash functions, \TLSS, and KhPRF functions. Then, the \NFSA protocol described in \cref{alg:offline} and \cref{alg:online} does not leak $\{\mathbf{x}_i^{(r)}\}_{r \in [1:R]}$ to the server, except for the aggregated results $\{\mathbf{x}^{(r)}\}_{r\in[1:R]}$ under the semi-honest model.
\end{theorem}

The formal proof is given in \cref{apx:proof_of_security_nfsa}.

\begin{theorem}
    The result of \cref{thm:security_nfsa} holds even when the server colludes with fewer than $T$ decryptors.
\end{theorem}

\begin{proof}
    Note that during the execution of the \NFSA protocol, no messages are forwarded from the server to the decryptors. Therefore, when the server colludes with fewer than $T$ decryptors, the only information they can obtain is the \TLSS share values. According to \cref{thm:ss_security_collusion}, when the server colludes with fewer than $T$ decryptors, they cannot learn anything about the secret, i.e., they gain no information about $\mathbf{k}_i^{(r)}$. Consequently, their view is identical to the case where there is no collusion. Thus, \cref{thm:security_nfsa} holds.
\end{proof}

Furthermore, our scheme is clearly secure against collusion attacks involving only decryptors. By the security of \TLSS, all the $\{\mathbf{s}_m^{A_1,r}\}_{m \in \mathcal{P}_2}$ are indistinguishable from true random values. Additionally, by the security of the hash function, these values are independent of each other. Therefore, even if the adversary obtains the views of all decryptors, they cannot gain any information about the users' private inputs.

We now prove that KhPRF is secure in our application scenario. We  have the following theorem:
\begin{theorem}
    Even when the server has $\mathbf{k}^{(r)}$, it cannot distinguish $\mathbf{c}_i^{(r)}$ from a uniformly distributed vector of the same length in $\mathbb{Z}_{\Delta p_c}$.
\end{theorem}

\begin{proof}
    We aim to show that the distributions $(\mathbf{c}_i^{(r)}, \mathbf{k}^{(r)})$ and $(\mathbf{c}_i^\prime, \mathbf{k}^{\prime(r)})$ are indistinguishable, where $\mathbf{c}_i^\prime$ is a uniformly distributed random vector in $\mathbb{Z}_{\Delta p_c}$ of the same length as $\mathbf{c}_i^{(r)}$, and $\mathbf{k}^{\prime(r)}$ is a uniformly distributed random vector in $\mathbb{Z}_q^{n_\lambda}$.

    Since each $\mathbf{k}_i^{(r)}$ is  uniformly sampled from $\mathbb{Z}_q^{n_\lambda}$, $\mathbf{k}^{(r)}$ is indistinguishable from a uniformly distributed random vector in $\mathbb{Z}_q^{n_\lambda}$. Hence, $(\mathbf{c}_i^{(r)}, \mathbf{k}^{(r)})$ is indistinguishable from $(\mathbf{c}_i^{(r)}, \mathbf{k}^{\prime(r)})$.
     
    Assume, for the sake of contradiction, that there exists a distinguisher $\mathcal{A}$ capable of distinguishing $(\mathbf{c}_i^{(r)}, \mathbf{k}^{(r)})$ from $(\mathbf{c}_i^{(r)}, \mathbf{k}^{\prime(r)})$. We can use $\mathcal{A}$ to construct a distinguisher $\mathcal{B}$ that can distinguish between $F_H(\mathbf{k}_i^{(r)},r)$ and a same length uniformly random vector in $\mathbb{Z}_{\Delta p_c}$. The input to $\mathcal{B}$ is $\mathbf{r}_B$. $\mathcal{B}$ then samples $\mathbf{k}^{\prime(r)}$ randomly and constructs $\mathbf{c}_B = \Delta \times \operatorname{CRT}(\mathbf{x_i}^{(r)}) - \mathbf{r}_B \mod \Delta p_c$. It provides $(\mathbf{c}_B, \mathbf{k}^{\prime(r)})$ to $\mathcal{A}$, and returns $\mathcal{A}$'s output.
    From \cite{boneh2013key}, $F_H(\mathbf{k}_i^{(r)}, r)$ and the uniform distribution over $\mathbb{Z}_{\Delta p_c}$ are indistinguishable. Thus, no such distinguisher $\mathcal{A}$ exists.

    Clearly, $(\mathbf{c}_i^{(r)}, \mathbf{k}^{\prime(r)})$ and $(\mathbf{c}_i^\prime, \mathbf{k}^{\prime(r)})$ are also indistinguishable. This is equivalent to $F_H(\mathbf{k}_i^{(r)}, r)$ being indistinguishable from $\Delta \times \operatorname{CRT}(\mathbf{x_i}^{(r)}) - \mathbf{c}_i^\prime \mod \Delta p_c$. Since $\mathbf{c}_i^\prime$ is uniformly distributed, $\Delta \times \operatorname{CRT}(\mathbf{x_i}^{(r)}) - \mathbf{c}_i^\prime \mod \Delta p_c$ is also uniformly distributed. By the security of $F_H$, these distributions are indistinguishable.
    
    By transitivity of indistinguishability, we conclude that $(\mathbf{c}_i^{(r)}, \mathbf{k}^{(r)})$ and $(\mathbf{c}_i^\prime, \mathbf{k}^{\prime(r)})$ are indistinguishable.
\end{proof}

\textbf{Note:} The exposure of $\mathbf{k}^{(r)}$ allows the server to compute the error introduced by the aggregated KhPRF masking. However, this error does not reveal any user's input. Otherwise, it could be used to construct a distinguisher between $\mathbf{c}_i^{(r)}$ and a uniformly random distribution, which would contradict the above theorem.

\subsection{Discussion on Malicious Security}

Malicious security is critical in many scenarios. Although our current solution does not implement malicious security, we will discuss possible approaches to achieve it and the remaining challenges. In FL, malicious threats primarily come from two sources: malicious server and malicious users.

\subsubsection{Malicious Server}
When the server is malicious, it may forge user inputs to deceive the decryptors or return incorrect aggregation results. OPA \cite{karthikeyan2025opa} embeds signature information into a mask to ensure that a malicious server cannot forge other users' inputs. This approach can be directly applied to our scheme, as the sharing of signature information and private keys is identical. Thus, it can be used with \TLSS. The additional overhead is that each user must perform one extra \TLSS share, and both the decryptors and the server incur one more reconstruction overhead.

OPSA \cite{guan2025opsa} uses the Pedersen commitment \cite{pedersen1991non} to detect incorrect aggregation results. Users commit to their inputs using the commitment protocol and leverage the homomorphic property of the commitment, enabling the server to generate proofs for the aggregation results. Users then verify the global model. By adding the aggregation round information into the verification tag, each round's aggregation result becomes unique. According to \cite{guan2025opsa}, the additional overhead is approximately 2 to 3 times. Since our CRT packing method reduces the model's dimensionality, the verification overhead is also reduced, approximately $O(\frac{1}{d_c})$, compared to directly using the method in \cite{guan2025opsa}.

\subsubsection{Malicious Users}
When one user is malicious, there are two possible malicious behaviors: dishonest key sharing and masking, or encrypting malicious update results. To resist malicious users, we need a robust \TLSS. Below, we discuss how to construct a robust \TLSS from verifiable Shamir's SS \cite{feldman1987practical}.

In verifiable Shamir's SS, when sharing, a commitment $c_m = g^{s_m}$ is provided for each share value, where $g$ is the generator of a group of order $p$. In \TLSS, since $s_m$ is not directly sent to $P_m$, it cannot be verified by $P_m$ without assistance from the server. The server computes $c_m' = g^{s_m^{A_2}}$. Since $s_m = s_m^{A_1} + s_m^{A_2}$, and $P_m$ possesses $s_m^{A_1}$, $P_m$ can compute $c_m'' = g^{s_m^{A_1}}$. Then, $P_m$ verifies whether $c_m = c_m' \cdot c_m''$ holds to verify $s_m$. The required overhead for the user is one exponentiation and the transmission of a group element. The server and the share holders both incur an additional exponentiation. For vector verification, the optimization from \cite{guan2025opsa} can be used, which reduces the communication overhead to that of committing to a single element.

For verification of the reconstruction, two methods can be used: one through commitments and the other through Reed-Solomon decoding. Share holders can compute the commitment of the aggregated share values using the homomorphic property of the commitment and send it to the server. After the server reconstructs the second-layer shares, it obtains the Shamir's SS shares and can verify the commitment to check if the result returned by $P_m$ is correct. The server requires $M$ exponentiations to verify each share, while the share holder requires $N$ multiplications. Alternatively, the server can use Reed-Solomon decoding to reconstruct Shamir's SS shares, in which case the share holder does not need additional computation. However, Reed-Solomon decoding requires that the number of malicious shares be fewer than $\frac{M}{3}$. Using the commitment protocol only requires $T$ honest shares.

The verification of KhPRF can utilize the ZKP techniques from \cite{karthikeyan2025opa}. The Proof of Correct Masking and the Proof of Linear Relations are similar. Therefore, the additional overhead for the verification without considering CRT packing is the overhead of the verifiable \TLSS plus the overhead of the Proof of Correct Masking and Proof of Linear Relations.

However, a challenge remains in achieving a fully verifiable \NFSA. In our approach, we can only implement a coarse-grained input verification. By using CRT packed inputs as the verification objects, we can directly apply the verification method from \cite{karthikeyan2025opa}, with an overhead similar to that of the method in \cite{karthikeyan2025opa}. However, it is difficult to verify the range of the values within the packed CRT elements. This is because the CRT encoding process mixes different plaintext values, making the verification difficult. Future research could explore how ZKP techniques can be used to implement verifiable CRT encoding and decoding.

\section{Evaluation}
In this section, we implement our protocol using Python and conduct a series of numerical simulations to validate the effectiveness of our method. Since our focus is on the overhead of secure aggregation, without altering the model's training and aggregation rules, we mainly use the following two experimental metrics:

\begin{itemize}
    \item \textbf{Traffic}: The size of the serialized data. Since we focus on bandwidth-limited user scenarios, \textit{Traffic} is the size of the data obtained after serializing with Python's \texttt{pickle} library and compressing with Python's \texttt{gzip} library.
    \item \textbf{Time}: The total computation time required for data processing. To better reflect the real computational cost, \textit{Time} includes the computation time, serialization time, and compression processing time.
\end{itemize}

All of our metrics are based on simulation results on an Intel Xeon Gold 6145 CPU, with averages taken over ten independent trials. We use the SHA-256 hash function. The PRF function is implemented using AES with CTR mode using \textit{randomgen} Python library. The KA protocol is instantiated with elliptic curve Diffie-Hellman (ECDH). The security parameters for all cryptographic primitives are instantiated with at least 128-bit security.

\subsection{Performance of \TLSS}
To better illustrate the performance of \TLSS in sharing and reconstruction, we assume that a user needs to share the secret (a vector) with multiple holders. Our baseline is the classic Shamir's SS scheme (referred to as Shamir). The default threshold for Shamir is 3. As analyzed in OPA \cite{karthikeyan2025opa}, packed Shamir requires the number of holders to satisfy certain conditions. Therefore, to validate our scheme over a broader range, we implement the original Shamir scheme. Additionally, since the key exchange part is identical in both schemes, we do not consider the overhead of it.

\cref{tab:tlss_performance} shows a comparison of the performance between \TLSS and Shamir when the number of holders is 50 and the length of the secret is 5000, where $\lceil \log p \rceil$ denotes the bit length of the modulus. It can be observed that the user overhead of Shamir is comparable to that of \TLSS. The user time overhead of Shamir is slightly lower than that of \TLSS, while the user traffic of \TLSS is slightly lower than that of Shamir. This is because Shamir uses authenticated encryption to ensure the security of the relay, which requires additional verification information, whereas \TLSS does not.

The main advantage of \TLSS is the reduction in holder overhead. When the bit length of the modulus is 64, the holder computation time of \TLSS is reduced by approximately 95\%, and the communication overhead is reduced by about 57\%. This is because \TLSS does not require processing the relay data of users during sharing. In terms of total overhead (the sum of all entities' overhead), \TLSS is also significantly lower than Shamir. 

It is important to note that when the modulus size increases, the computation overhead for \TLSS holders increases noticeably. This is because our implementation of the PRF is not highly optimized. Since the \textit{randomgen} Python library we used does not directly support the generation of large integers, we implemented it by first generating random smaller integers (e.g., 64-bit) and then concatenating them. However, despite this, the overhead of \TLSS remains lower than that of Shamir.

In \cref{fig:m63_holder_performance}, we present the variation in holder overhead for different secret lengths. As the secret length increases, both the holder computation overhead and communication overhead grow linearly. However, the growth rate of \TLSS is noticeably slower than that of Shamir. As shown in \cref{fig:m63_dix_total_overhead}, the total overhead of \TLSS also increases more slowly than that of Shamir, demonstrating that \TLSS has better scalability. Furthermore, in \cref{fig:m63_dih_overhead}, we compare the overall overhead for different numbers of holders, where \TLSS again shows a slower increase in overhead compared to Shamir. This is consistent with our theoretical overhead analysis in \cref{tab:tlss_theory_compare}.

Additionally, we provide a comparison of the overhead when the modulus bit length is 128 in \cref{apx:m128}.

\begin{table}[!t]
\centering
\caption{Comparison of Shamir and \TLSS with 50 holders and  secret length 5000.}
\label{tab:tlss_performance}
\begin{tabular}{@{}c@{\hspace{0.1em}}c@{\hspace{0.1em}}cccc@{}}
\toprule
 & \multirow{2}{*}{Scheme} & \multicolumn{2}{c}{$\lceil \log p \rceil =64$} & \multicolumn{2}{c}{$\lceil \log p \rceil =128$} \\ \cmidrule(l){3-6} 
 &  & Time (s) & Traffic (MB) & Time (s) & Traffic (MB) \\ \midrule
\multirow{2}{*}{User} & Shamir & 1.24 & 2.24 & 1.28 & 4.20 \\
 & \TLSS & 1.51 & 2.23 & 1.48 & 4.19 \\ \midrule
\multirow{2}{*}{Holder} & Shamir & 0.022 & 0.090 & 0.216 & 0.168 \\
 & \TLSS & 0.001 & 0.038 & 0.008 & 0.085 \\ \midrule
\multirow{2}{*}{Total} & Shamir & 2.53 & 6.72 & 2.64 & 12.60 \\
 & \TLSS & 1.80 & 4.15 & 2.17 & 8.41 \\ \bottomrule
\end{tabular}
\end{table}

\begin{figure}[!t]
	\centering  
    \subfloat[Holder Time]{
    \label{fig:m63_dix_holder_time}
    \includegraphics[width=0.45\linewidth]{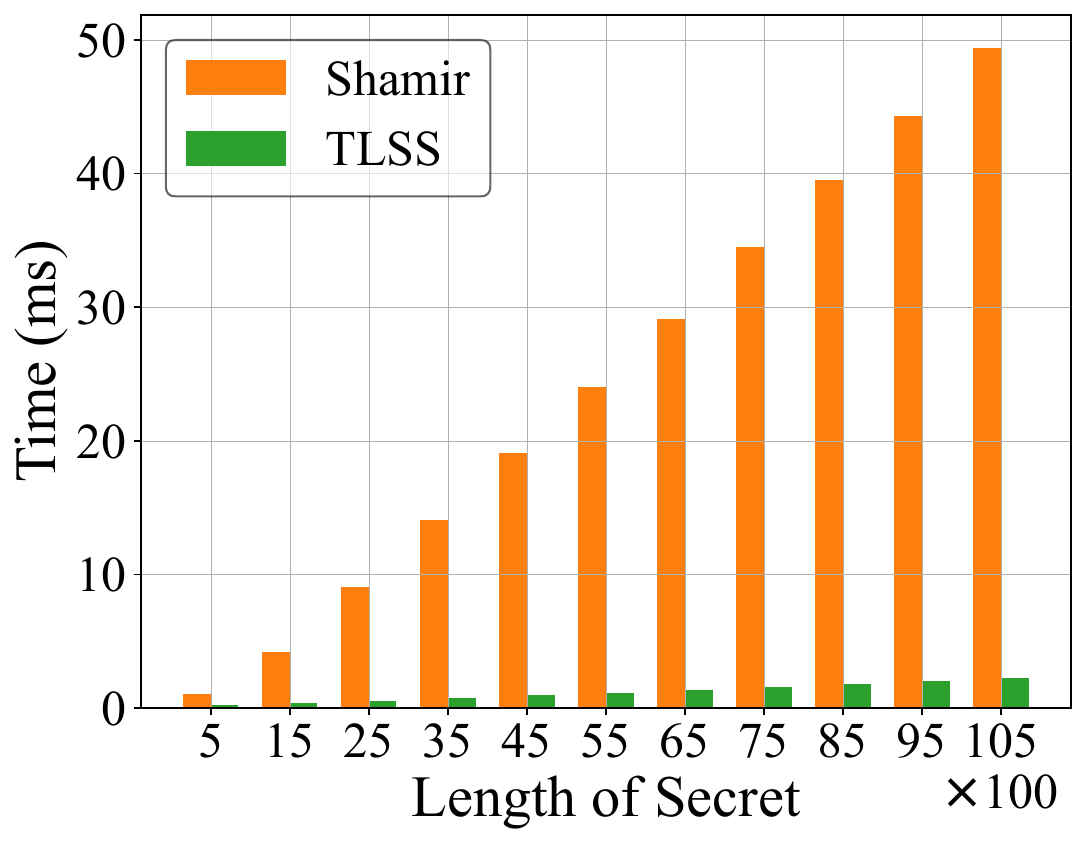}}
     \subfloat[Holder Traffic]{
    \label{fig:m63_dix_holder_traffic}
    \includegraphics[width=0.45\linewidth]{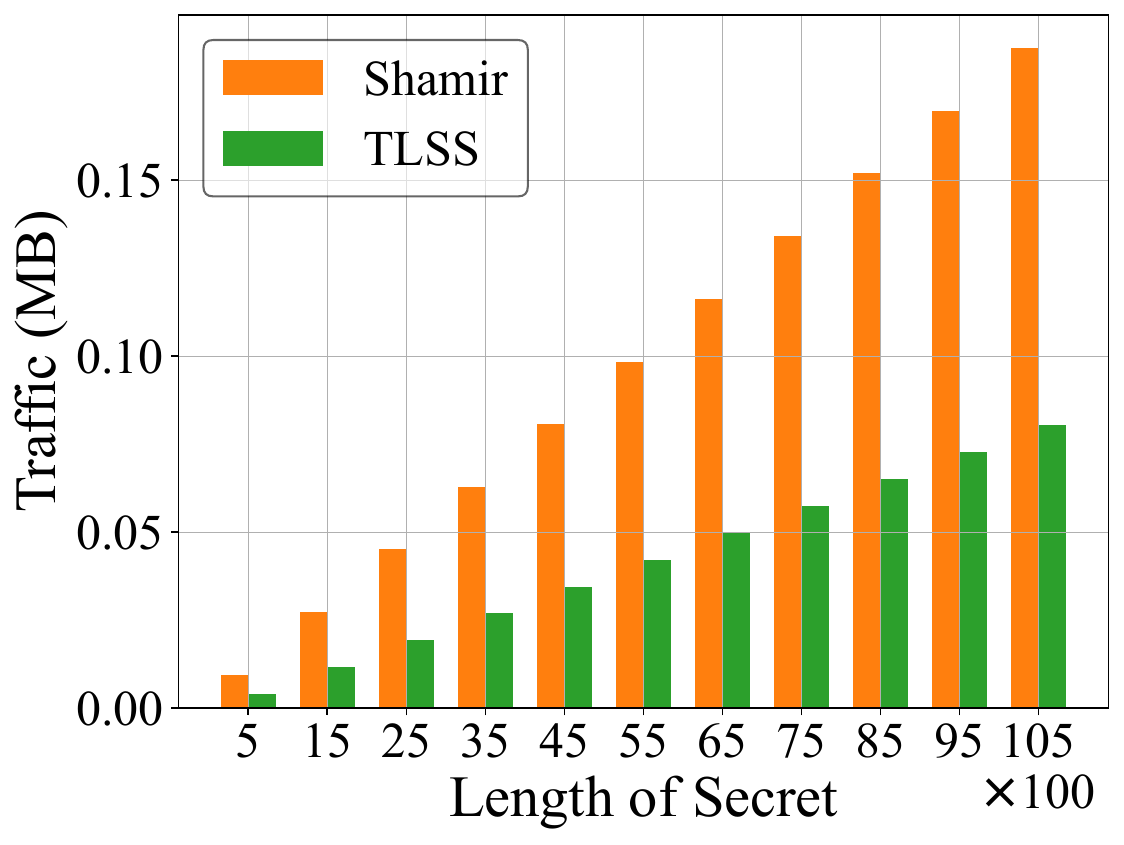}}
    
\caption{Holder performance with 50 holders and modulus bit length 64.}
\label{fig:m63_holder_performance}
\end{figure}

\begin{figure}[!t]
	\centering  
    \subfloat[Total Time]{
    \label{fig:m63_dix_sim_time}
    \includegraphics[width=0.45\linewidth]{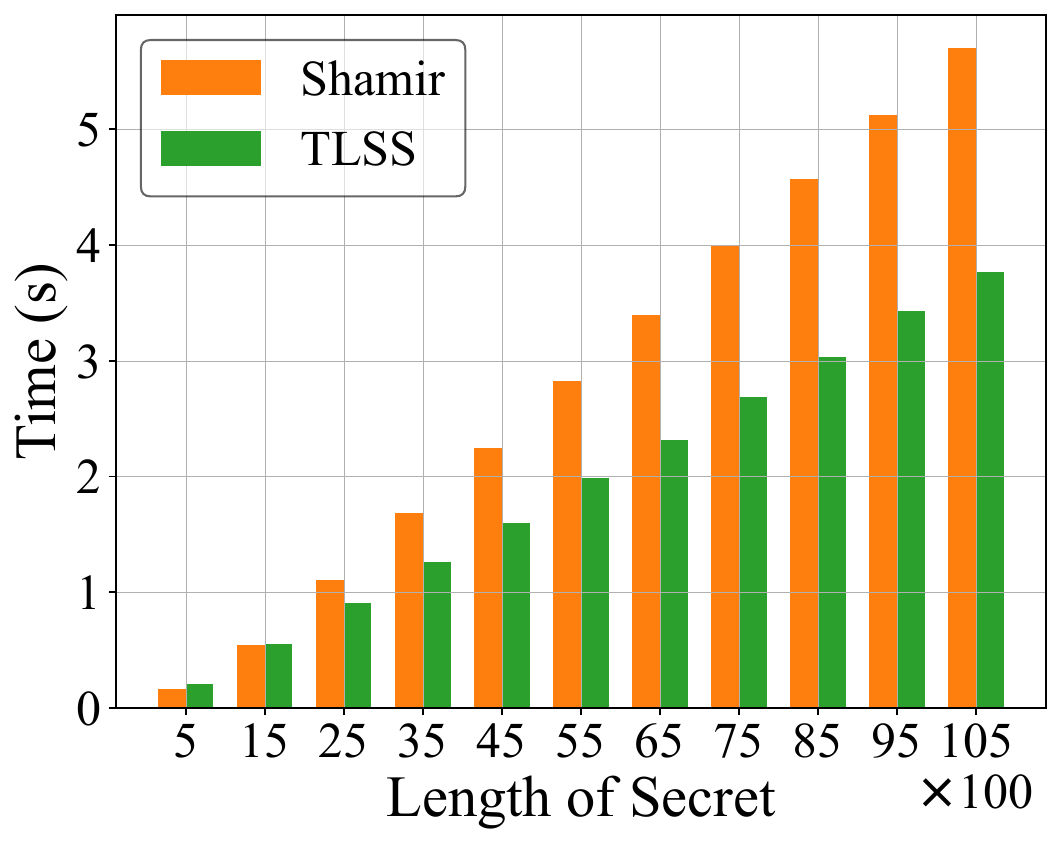}}
     \subfloat[Total Traffic]{
    \label{fig:m63_dix_server_traffic}
    \includegraphics[width=0.45\linewidth]{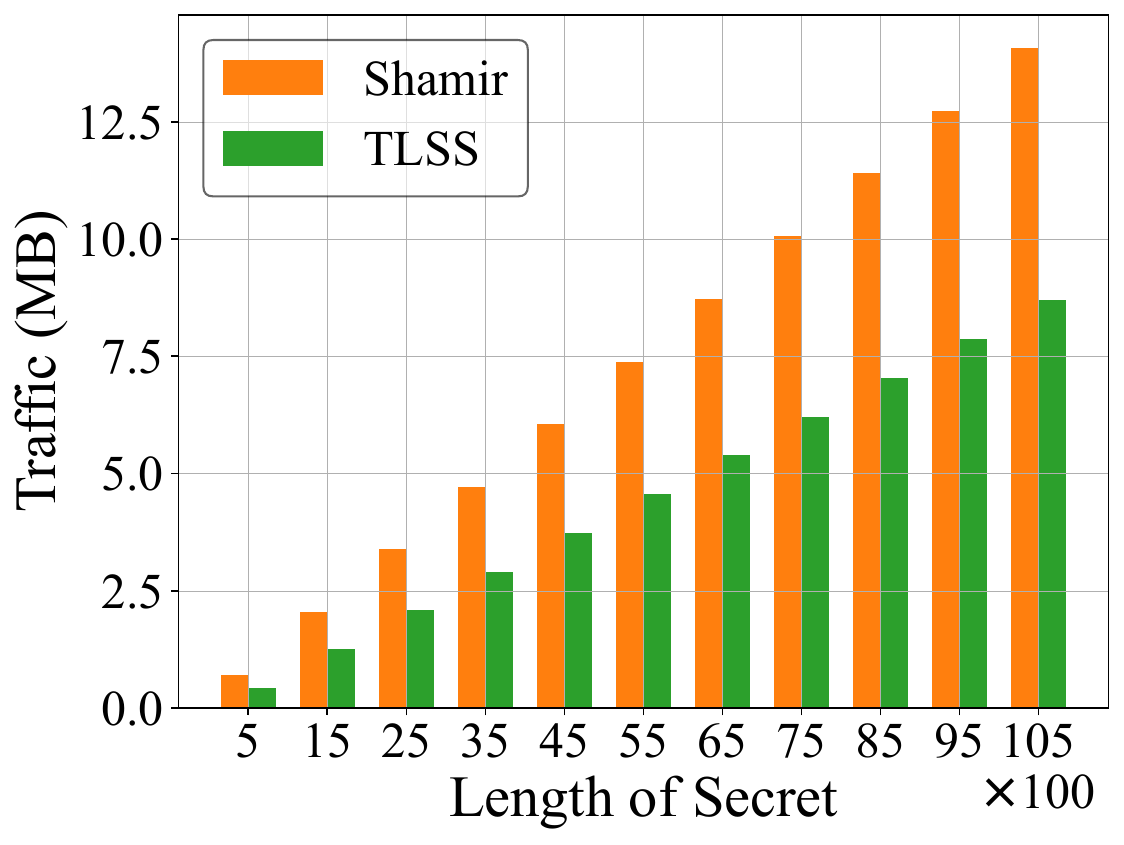}}
    
\caption{Comparison of total overhead for different secret lengths with 50 holders and  modulus bit length 64.}
\label{fig:m63_dix_total_overhead}
\end{figure}

\begin{figure}[!t]
	\centering  
    \subfloat[Total Time]{
    \label{fig:m63_dih_sim_time}
    \includegraphics[width=0.45\linewidth]{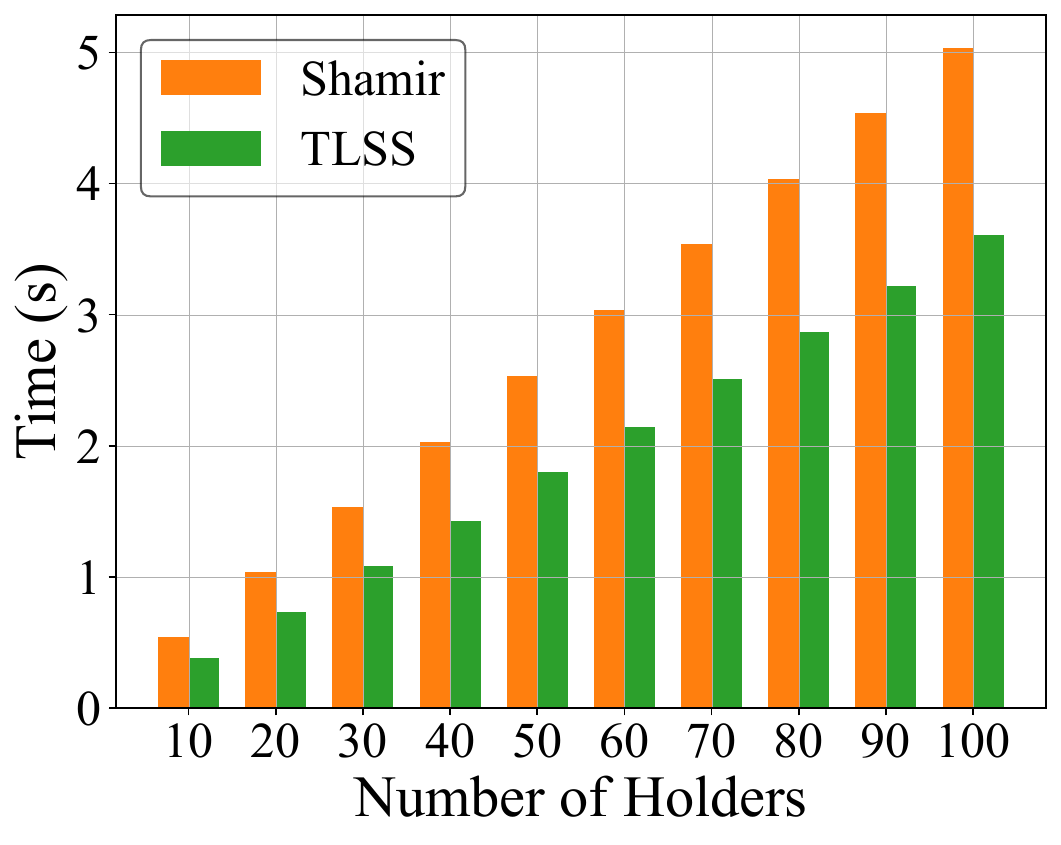}}
     \subfloat[Total Traffic]{
    \label{fig:m63_dih_server_traffic}
    \includegraphics[width=0.45\linewidth]{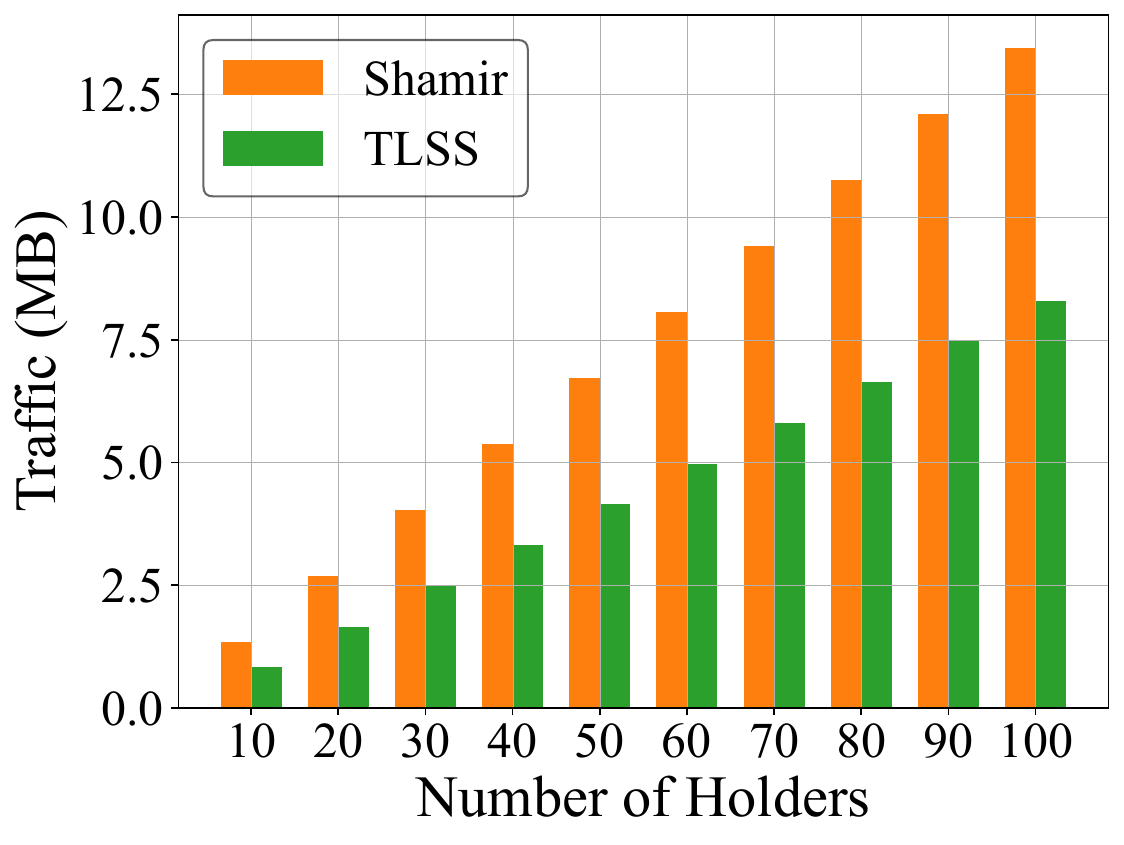}}
    
\caption{Comparison of total overhead for different numbers of holders with secret length 5000 and modulus bit length 64.}
\label{fig:m63_dih_overhead}
\end{figure}

\subsection{Performance of CRT Encoding}
Another key contribution of our work is the use of CRT to encode the input for KhPRF. Therefore, we evaluate the impact of different batch sizes ($d_c$) for CRT encoding. To ensure that KhPRF provides at least 128-bit security, we set the polynomial ring dimension to 8192, and the ciphertext modulus to a prime with a 180-bit length. The plaintext modulus is set to $2^{32}$ to meet the precision requirements of most FL scenarios. We set $d_c$ to 2, 3, and 4. Our baseline is the masking method from OPA \cite{karthikeyan2025opa}. 
Since our scheme focuses on the semi-honest setting, the comparison here is made with the semi-honest version of OPA, without considering any verification overhead.
When $d_c = 1$, our method is equivalent to the OPA masking method. We use speedup$_t$ and speedup$_c$ to denote the time and traffic speedups of our scheme relative to OPA. The default value for $\Delta$ is set to $2^{10}$ (note that $\Delta$ is larger than the maximum number of users participating in the summation), and the default input length is set to $2^{18}$.

\begin{table}[!t]
\centering
\caption{Masking performance with CRT encoding for $\Delta = 2^{10}$ and input length $2^{18}$.} 
\label{tab:crt_encoding_performance}
\begin{tabular}{@{}c@{\hspace{0.7em}}ccccc@{}}
\toprule
\multicolumn{2}{c}{Scheme} & Time (s) & Speedup$_t$ & Traffic (MB) & Speedup$_c$ \\ \midrule
\multicolumn{2}{c}{OPA \cite{karthikeyan2025opa}} & 7.56 & 1.00$\times$ & 1.67 & 1.00$\times$ \\ \midrule
\multirow{3}{*}{Ours} & $d_c=2$ & 3.90 & 1.94$\times$ & 1.35 & 1.23$\times$ \\
 & $d_c=3$ & 2.70 & 2.80$\times$ & 1.25 & 1.34$\times$ \\
 & $d_c=4$ & 2.03 & 3.72$\times$ & 1.19 & 1.40$\times$ \\ \bottomrule
\end{tabular}
\end{table}

As shown in \cref{tab:crt_encoding_performance}, our CRT encoding significantly reduces the overhead of masking for users. The speedup achieved by CRT encoding is close to $d_c$, because the time for CRT encoding itself is very fast. As a result, the masking time primarily depends on the time taken by KhPRF. After CRT encoding, the input is reduced to $\frac{1}{d_c}$ of its original size, so the number of KhPRF calls required also reduces by a factor of $\frac{1}{d_c}$. The communication gain for users mainly comes from the reduction in the total number of $\Delta$ values after encoding. However, since $\Delta$ does not constitute a large portion of the overall communication, the speedup in communication overhead is less than $d_c$.

\begin{figure}[!t]
	\centering 
    \subfloat[Masking Time]{
    \label{fig:crt_input_length_time}
    \includegraphics[width=0.45\linewidth]{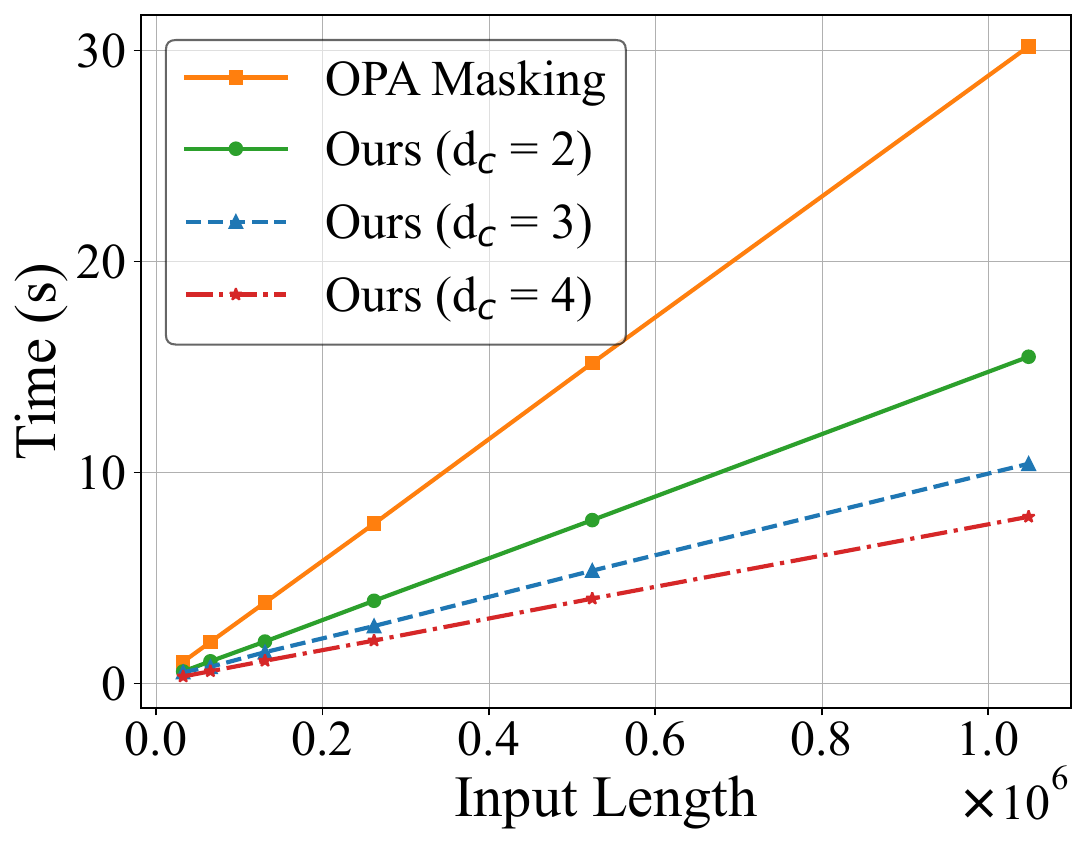}}
     \subfloat[Traffic]{
    \label{fig:crt_input_length_traffic}
    \includegraphics[width=0.45\linewidth]{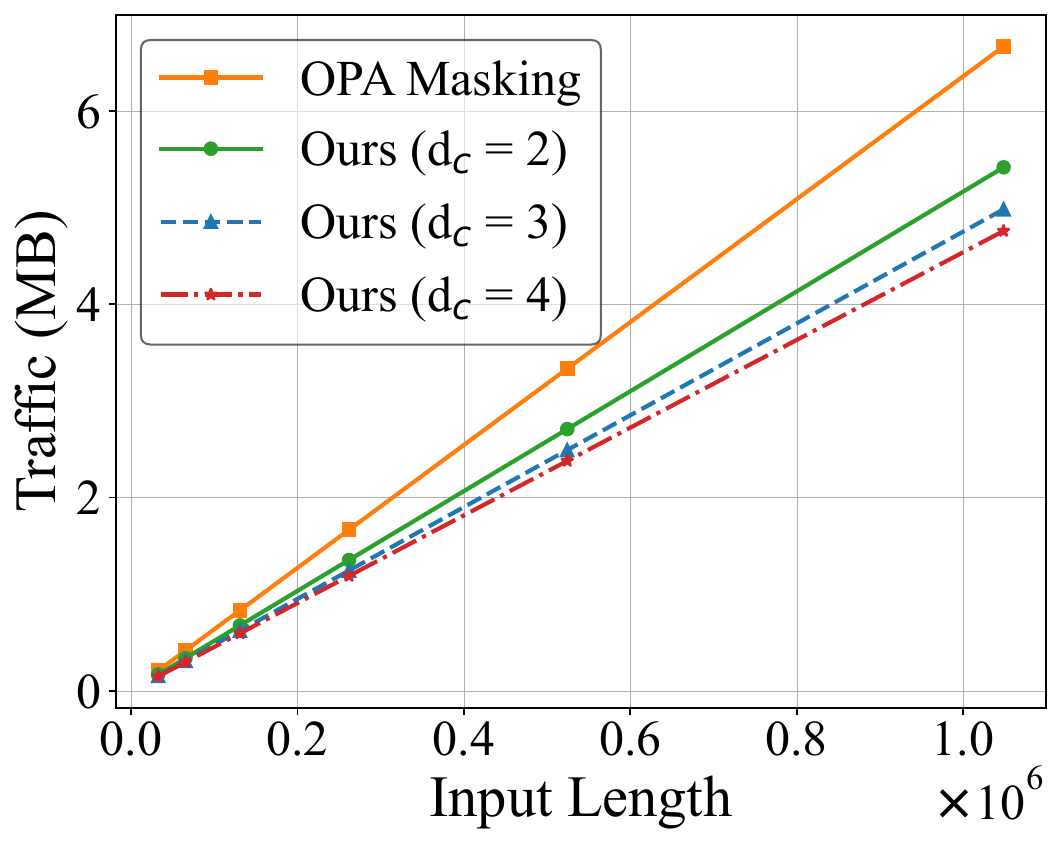}}
    
\caption{Overhead comparison of KhPRF-based masking under different input lengths with $\Delta = 2^{10}$. Masking time includes CRT encoding and KhPRF computation.}
\label{fig:crt_input_length}
\end{figure}

\begin{figure}[!t]
	\centering  
    \subfloat[Masking Time]{
    \label{fig:crt_num_user_time}
    \includegraphics[width=0.45\linewidth]{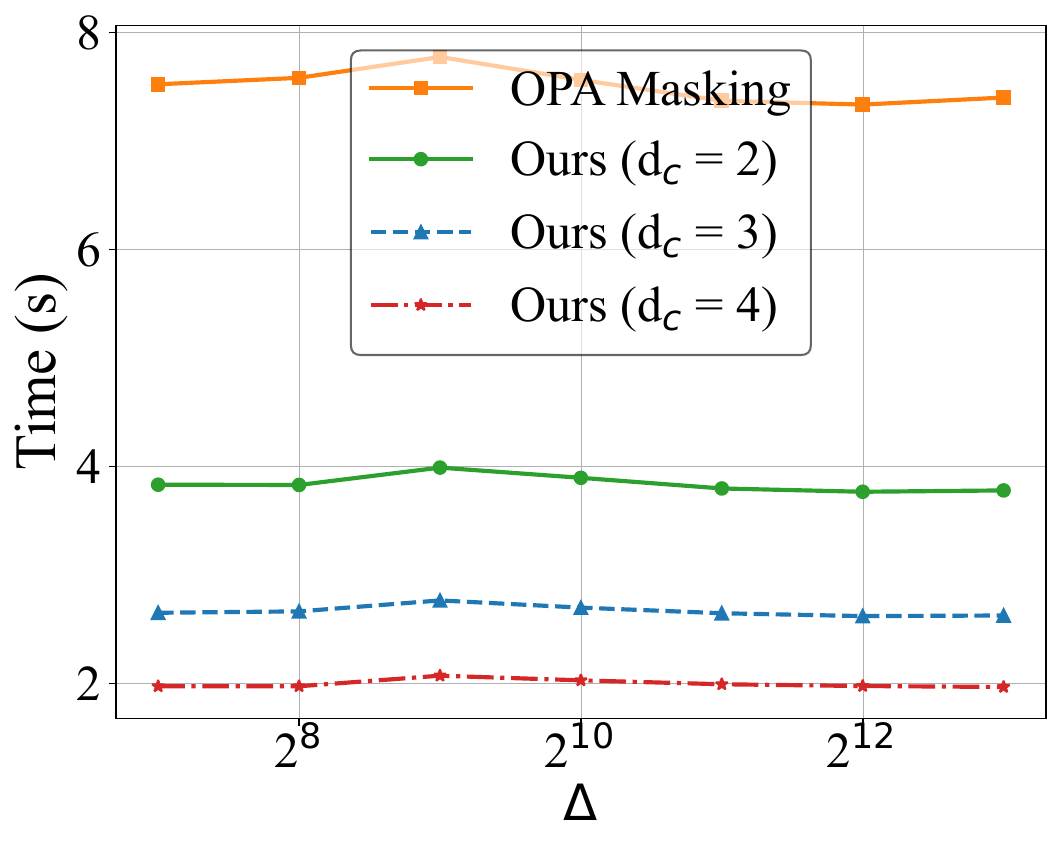}}
     \subfloat[Traffic]{
    \label{fig:crt_num_user_traffic}
    \includegraphics[width=0.45\linewidth]{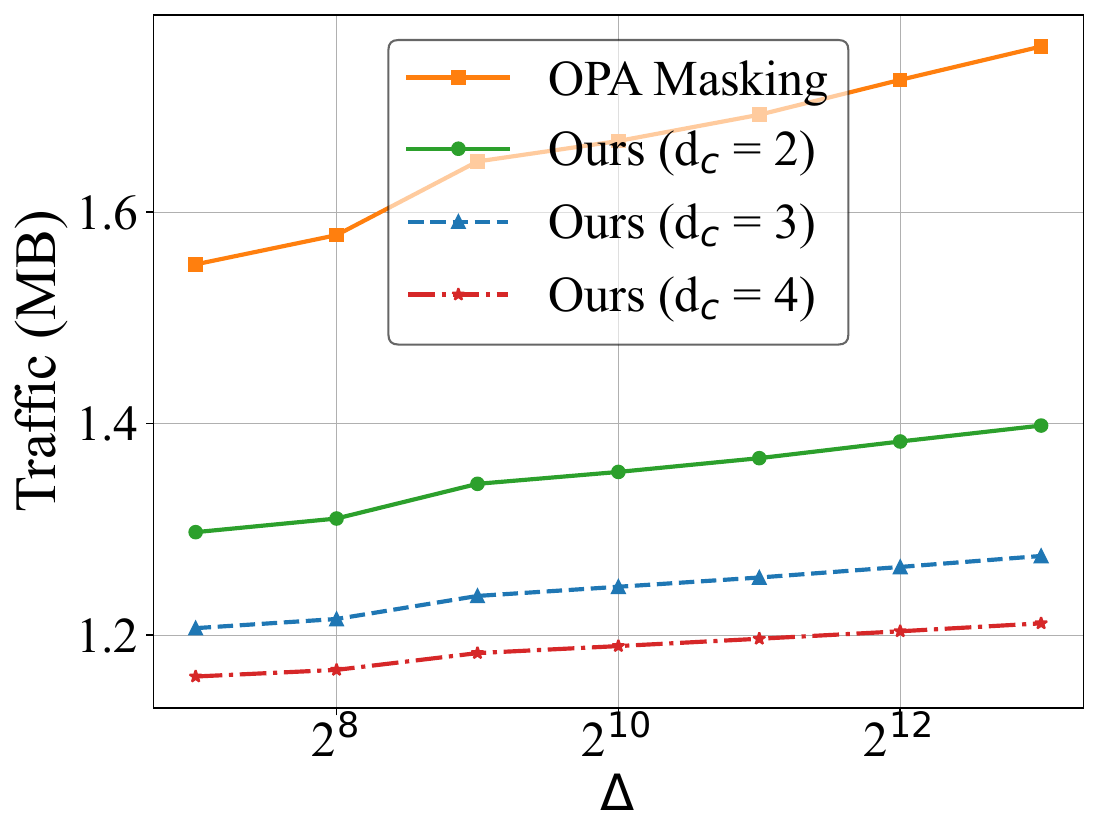}}
    
\caption{Overhead comparison of KhPRF-based masking under different $\Delta$ values with input length $2^{18}$. The x-axis is logarithmic with base 2.}
\label{fig:crt_num_user}
\end{figure}

In \cref{fig:crt_input_length}, we compare the time required for masking and the communication volume after masking for different input lengths. It is evident that the reduction in masking time is quite significant. The communication overhead shows a noticeable decrease when $d_c = 2$, while for $d_c = 3$ and $d_c = 4$, the communication overhead only decreases slightly. We also present the change in overhead as $\Delta$ increases in \cref{fig:crt_num_user}. From \cref{fig:crt_num_user_time}, it can be observed that as $\Delta$ grows, the masking time remains almost unchanged, because the computation time of KhPRF is not significantly affected by $\Delta$. The computation time mainly depends on the size of the ciphertext modulus and the polynomial ring dimension. As shown in \cref{fig:crt_num_user_traffic}, as $\Delta$ increases, the size of the final masked data grows, which makes the benefits of our encoding more pronounced.

\begin{table}[!t]
\caption{KhPRF masking overhead under different parameter configurations.}
\label{tab:user_different_parameters}

\begin{tabular}{@{}cccccc@{}}
\toprule
$\lceil \log p_m \rceil$ & $d_c$ & Mask (s) & Encode (ms) & Total (s) & Traffic (MB) \\ \midrule
4 & 32 & 0.32 & 51.47 & 0.37 & 0.21 \\
8 & 16 & 0.54 & 59.03 & 0.60 & 0.29 \\
16 & 8 & 1.05 & 66.73 & 1.12 & 0.56 \\
24 & 5 & 1.72 & 67.66 & 1.79 & 0.84 \\
32 & 4 & 2.01 & 80.20 & 2.09 & 1.12 \\ \bottomrule
\end{tabular}

\end{table}

\textbf{Parameters Setting and Impact of Different Sizes of Plaintext.} We fixed the polynomial ring dimension at $8192$, $\Delta = 2^{10}$, and $q$ as a prime number of $180$ bits, and compared the overhead of using KhPRF masking with different plaintext sizes when the input length is $2^{18}$, as shown in \cref{tab:user_different_parameters}. We used CRT packing to pack the plaintext, aiming for the largest possible packed bit length, but not exceeding 128 bits for security. The parameter settings and overheads for different plaintext sizes are shown in \cref{tab:user_different_parameters}. 
As shown, under the same security parameters, when $p_m$ is smaller, $d_c$ can be larger, which means fewer calls to KhPRF and consequently less time overhead for masking. When $d_c$ increases, the number of CRT encodings required decreases, but the time per encoding increases due to the larger $d_c$. Therefore, though the encoding time is reduced when $p_m$ is smaller, the difference is not significant. As $p_m$ decreases, the space required for each element becomes smaller, and as $d_c$ increases, the required $\Delta$ decreases, which results in reduced traffic.

\subsection{Performance of \NFSA}
To better demonstrate the performance of our \NFSA scheme in practical FL parameters, we evaluate the performance of secure aggregation on several models, which is consistent with most existing secure aggregation schemes, such as \cite{bell2023acorn,lycklama2023rofl}. The models tested are: (1) a CNN with 19K parameters on the Federated-MNIST dataset \cite{caldas2018leaf}, (2) a LeNet-5 \cite{lecun2002gradient} with 62K parameters, (3) a ResNet-20 with 273K parameters on the CIFAR-10 dataset \cite{krizhevsky2009learning}, and (4) an LSTM model \cite{graves2012long} with 818K parameters on the text-based Shakespeare dataset. Users locally optimize the models using stochastic gradient descent with learning rates of 0.05, 0.01, 0.05, and 0.3, respectively. Since our focus is on the performance of secure aggregation, we do not consider the time spent on model training or the communication volume for distributing the global model. All model parameters are represented using 32-bit values. We assume 100 users and 5 decryptors.

We compare our scheme with the latest secure aggregation framework, OPA \cite{karthikeyan2025opa}, which employs the same techniques. OPA uses Shamir's SS and KhPRF for masking user inputs without CRT encoding. For KhPRF, the polynomial ring dimension is set to 8192, and the ciphertext modulus is 180 bits. Our scheme uses CRT encoding with $d_c = 2, 3, 4$, referred to as Ours-2, Ours-3, and Ours-4, respectively. Note that the key sharing process is the same in both schemes, using \TLSS, so the overhead for key sharing remains identical across all approaches.

\begin{table}[!t]
\centering
\caption{Comparison of secure aggregation overhead on the Shakespeare dataset.}
\label{tab:nfsa_overhead}
\begin{tabular}{@{}cccccc@{}}
\toprule
 & Metric & OPA \cite{karthikeyan2025opa}& Ours-2 & Ours-3 & Ours-4 \\ \midrule
\multirow{2}{*}{User} & Time (s) & 26.17 & 12.77 & 8.91 & 6.51 \\
 & Traffic (MB) & 5.76 & 4.98 & 4.71 & 4.57 \\ \midrule
\multirow{2}{*}{Decryptor} & Time (s) & 0.65 & 0.54 & 0.54 & 0.54 \\
 & Traffic (MB) & 19.25 & 0.19 & 0.19 & 0.19 \\ \midrule
\multirow{2}{*}{Server} & Time (s) & 85.59 & 43.54 & 30.84 & 22.51 \\
 & Traffic (MB) & 672.16 & 499.29 & 471.64 & 457.65 \\ \bottomrule
\end{tabular}
\end{table}

We first present the experimental results on the Shakespeare dataset in \cref{tab:nfsa_overhead}, where our scheme consistently shows lower overhead compared to OPA. For users, the overhead consists of two parts: key sharing and input masking. The key sharing overhead is almost the same, but our scheme achieves lower masking time. As $d_c$ increases, the user's overhead decreases, demonstrating the effectiveness of CRT encoding. For decryptors, they only need to perform key reconstruction without processing any input-related information. Our scheme performs better, as \TLSS eliminates the server's intermediary communication. In OPA, each user's key share must be forwarded to the decryptors, resulting in a communication overhead approximately 100 times higher than ours. Specifically, OPA requires 19.25MB, while our scheme only requires 0.19MB. While the decryptor computation time in our scheme is much lower than in Shamir-based methods when using \TLSS alone, the aggregate time for key sharing becomes dominant. As a result, the overall decryptor computation overhead in our scheme is about 17\% lower than OPA. Additionally, our scheme significantly reduces the server's overhead: Ours-2 reduces server computation time by about 50\% and communication by about 25\% compared to OPA.

\begin{figure}[!t]
	\centering  
    \subfloat[User Time]{
    \label{fig:nfsa_user_time}
    \includegraphics[width=0.45\linewidth]{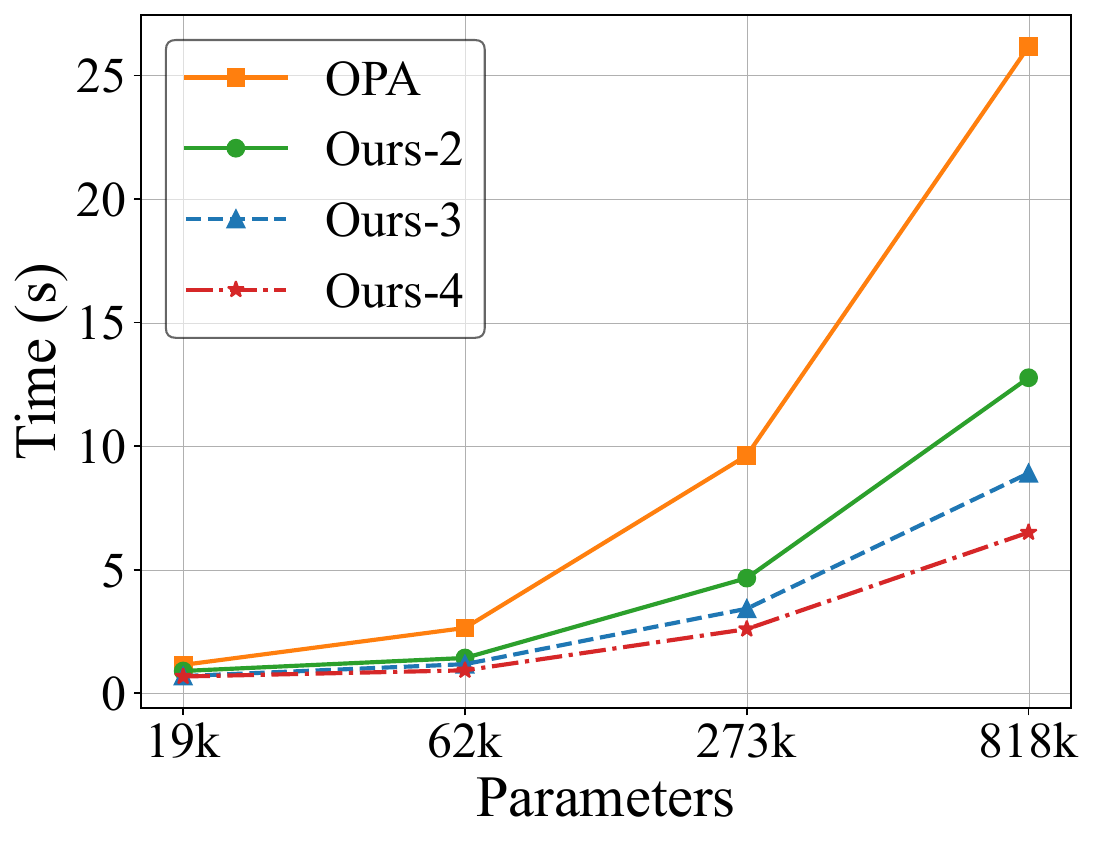}}
     \subfloat[User Traffic]{
    \label{fig:nfsa_user_traffic}
    \includegraphics[width=0.45\linewidth]{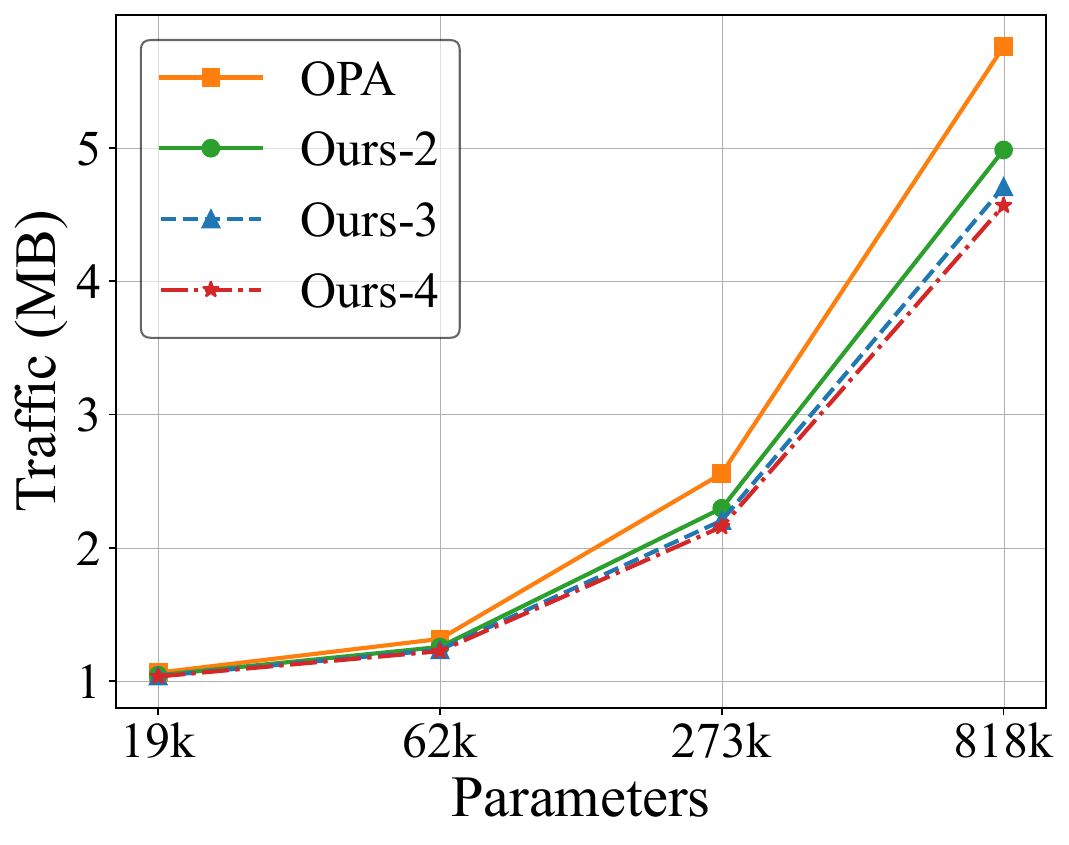}}
    
\caption{Comparison of user overhead for one round of aggregation.}
\label{fig:nfsa_user_overhead}
\end{figure}

\begin{figure}[!t]
	\centering  
    \subfloat[Server Time]{
    \label{fig:nfsa_server_time}
    \includegraphics[width=0.45\linewidth]{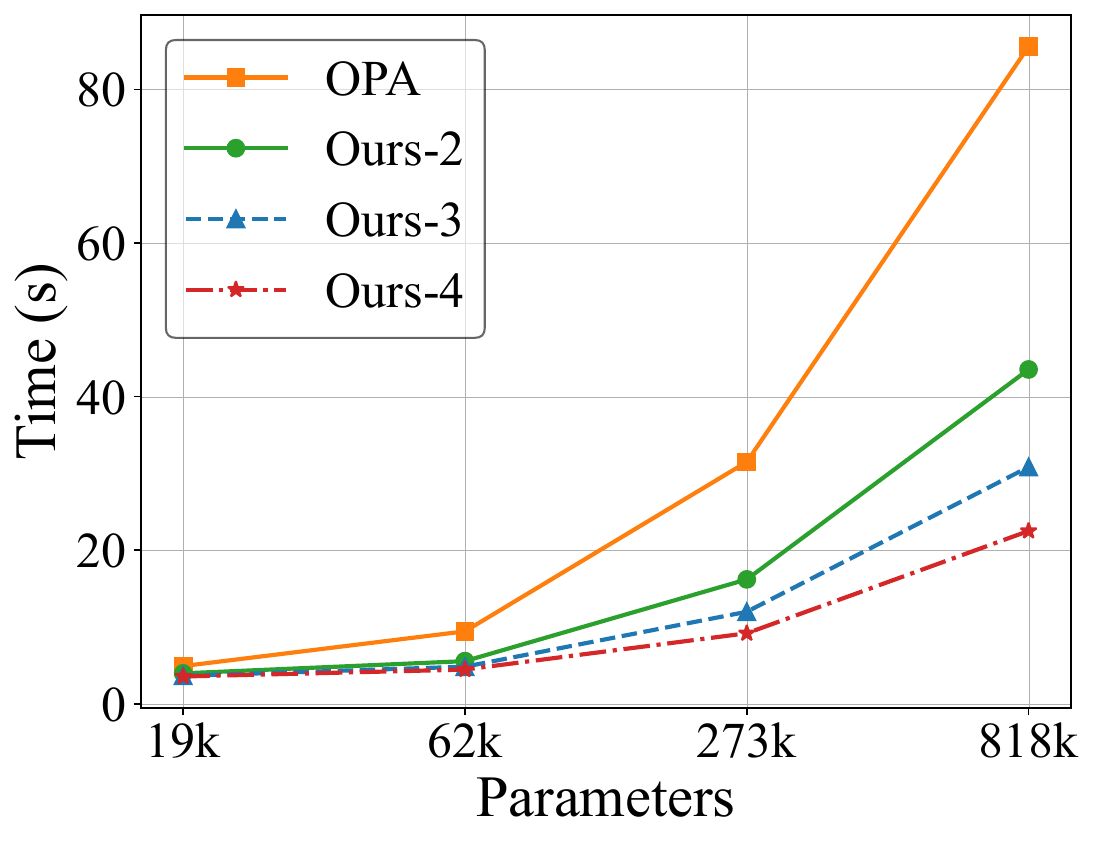}}
     \subfloat[Server Traffic]{
    \label{fig:nfsa_server_traffic}
    \includegraphics[width=0.45\linewidth]{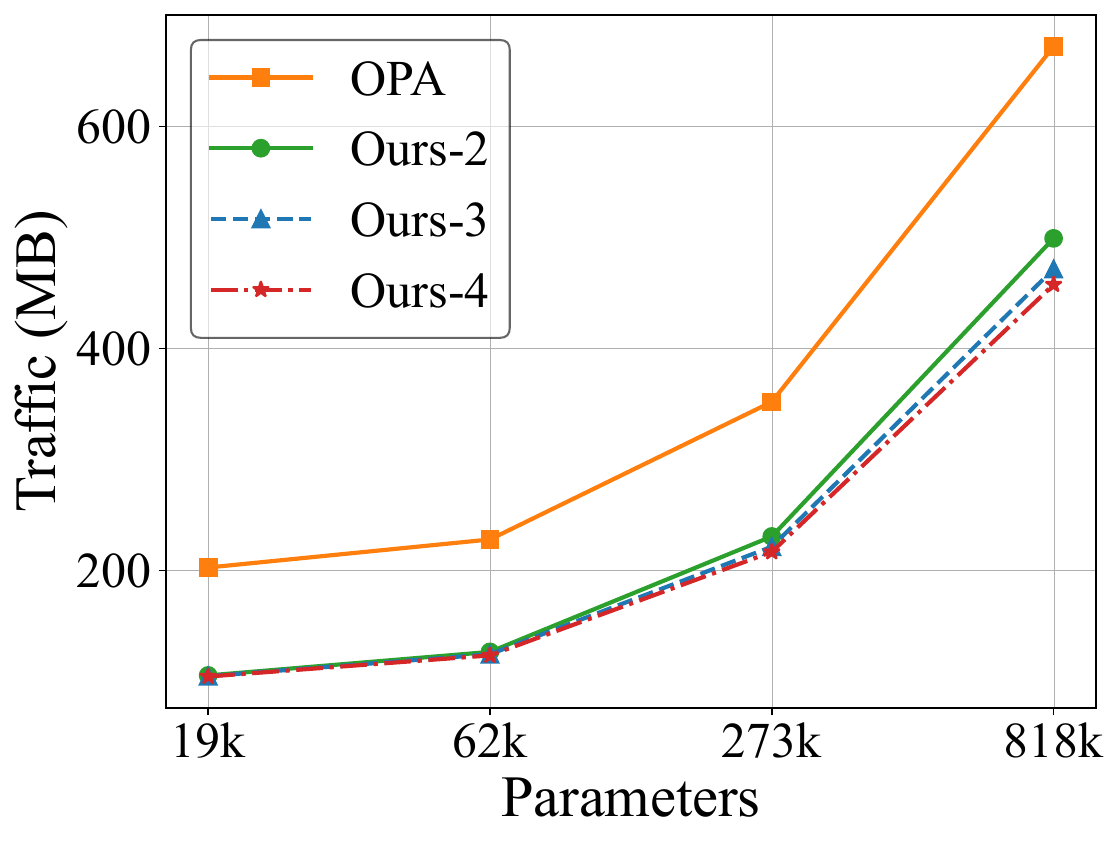}}
    
\caption{Comparison of server overhead for one round of aggregation.}
\label{fig:nfsa_server_overhead}
\end{figure}

Since the decryptor's overhead is independent of the model, we only compare the user and server overhead in \cref{fig:nfsa_user_overhead} and \cref{fig:nfsa_server_overhead}. As shown in \cref{fig:nfsa_user_time}, our scheme demonstrates greater advantages as the model size increases. When the model parameters are 19K and 62K, the user communication overhead of our scheme is almost the same as OPA's, but our scheme requires less user time. The server time overhead, shown in \cref{fig:nfsa_server_time}, follows a similar trend. However, in \cref{fig:nfsa_server_traffic}, the server traffic of OPA is consistently much higher than that of our scheme, due to \TLSS, which significantly reduces the communication overhead for key sharing in our approach.

Overall, our proposed secure aggregation scheme, \NFSA, achieves lower time and communication overheads compared to OPA, which uses the same underlying technologies, thanks to two key innovations: \TLSS and CRT encoding. In particular, our scheme significantly outperforms OPA in terms of user computation time, server communication, and decryptor computation and communication overhead.

\section{Conclusion}
In this work, we propose a novel SS scheme, \TLSS, specifically designed for the single-server FL scenario. This approach enables users to share secrets without requiring the server to relay secret information, thereby reducing the attack surface and overall communication overhead. Additionally, we introduce a new encoding scheme for input masking based on the almost KhPRF. By leveraging CRT, we batch multiple inputs together, which reduces the number of KhPRF calls and subsequently lowers the masking overhead for users. This also mitigates the issue of user communication inflation when many users participate in aggregation. Our theoretical analysis demonstrates that the proposed scheme is scalable and is practical for real-world applications. Although our experiments are conducted within the context of FL, the approach can be naturally extended to various secure aggregation tasks, such as electronic voting, smart-grid aggregation, and other secure aggregation tasks.
However, our scheme is designed for semi-honest environments, and for full malicious security, OPA remains a more robust alternative. Additionally, while CRT packing improves aggregation efficiency, it requires a larger modulus $q$ as the batch size increases, resulting in larger KhPRF key sizes. This also introduces added complexity for fine-grained input verification. 
In future work, we plan to further optimize the verifiability of \TLSS and \NFSA, enabling its use in efficient verifiable FL, thus expanding the scope of its applications.

\bibliographystyle{ACM-Reference-Format}
\bibliography{ref}

\appendix

\section{Ethical Considerations}
In this section, we discuss the potential ethical issues.

\textbf{Disclosures.} Our research does not involve any system or software vulnerabilities. 

\textbf{Experiments with live systems without informed consent.} Our research does not involve any live experiments. We focus on theoretical algorithms, and all experiments are conducted through computer simulations, with no live experimentation involved. From a scientific perspective, live experiments are not required for our study.

\textbf{Terms of service.} Our research did not involve any interaction with third-party services governed by terms of service.

\textbf{Deception. }  There is no deception in our research process. 

\textbf{Wellbeing for team members.} All team members involved in the research have collaborated according to the planned division of labor, and there were no significant disagreements during the research process. Our research does not include any disturbing or inappropriate material.

\textbf{Innovations with both positive and negative potential outcomes. } Our research provides users with improved secret sharing and FL, which significantly contribute to protecting user privacy and enhancing the user experience. At the same time, it offers service providers new opportunities for privacy-related services. The research has not caused any negative impact on other groups.

\textbf{Retroactively identifying negative outcomes.} Our solution does not involve any direct negative effects. Additionally, we have analyzed the drawbacks of our proposed solution, providing potential users with parameter recommendations and a detailed analysis of its advantages and disadvantages.

\textbf{The law.} Our research process and content fully comply with the law.

\section{Open Science}

This appendix outlines the artifacts provided to enable the evaluation of the core contributions of this paper. All necessary code is available through the GitHub repository at the following URL:

\textbf{Repository URL:}  
\url{https://github.com/pahjastia/NFSA-with-TLSS}

The core contributions of this paper include the \TLSS scheme and a CRT encoding method for almost KhPRF. The following artifacts are provided to evaluate these contributions:

\begin{itemize}
    \item \textbf{Code Artifacts:}
    \begin{itemize}
        \item \texttt{sim\_tlss.py}: Implements a single-threaded simulation of the TLSS protocol, including key generation, Shamir secret sharing, compression of shares, and secret reconstruction.
        \item \texttt{sim\_masking.py}: Implements a single-threaded simulation of the CRT encoding and KhPRF-based masking, focusing on vector aggregation and performance evaluation.
        \item \texttt{utils.py}: Contains helper functions for ECDH key agreement, Shamir secret sharing, AES-GCM encryption, and other utility functions.
        \item \texttt{dealer.py / holder.py / server.py}: Defines the roles and operations for dealers, holders, and servers in the \TLSS scheme.
       
        \item \texttt{crt\_coding.py}: Provides the necessary functions for CRT encoding and decoding, as well as related utilities.
        \item \texttt{KhPRF.py}: Implements the KhPRF based on RLWR with our NTT optimization.
    \end{itemize}
    
\end{itemize}

\section{More Preliminaries}
\subsection{Secure Aggregation for Federated Learning}
\label{apx:pre_sa}
In FL, the most commonly used aggregation rule is federated averaging \cite{mcmahan2017communication}. We represent model parameters as vectors. In each round of training, the new global model parameters $\mathbf{x}_0^{(r+1)}$ are the weighted average of the model parameters $\mathbf{x}_1^{(r)}, \mathbf{x}_2^{(r)}, \dots, \mathbf{x}_N^{(r)}$ from all users:
\begin{equation}
\mathbf{x}_0^{(r+1)} = \sum_{i=1}^N \alpha_i \mathbf{x}_i^{(r)},
\end{equation}
where the weight for each user is $\alpha_i = \frac{\Vert \mathcal{D}_i \Vert}{\sum_{i=1}^N \Vert \mathcal{D}_i \Vert}$.

In secure FL, we typically map the model parameters $\mathbf{x}_i$ of the users to a finite ring or finite field $\mathbb{Z}_p^d$ \cite{bonawitz2017practical}. Since model parameters are floating-point numbers, we usually apply quantization to convert them into integers. In this work, we use simple rounding quantization as in \cite{aono2017privacy}. Let $x$ be a real number and $\Delta_q$ the quantization factor. Then, the quantization operation is given by $\operatorname{Q}(x) = \lfloor x\Delta_q \rfloor$. The inverse quantization formula is simple: to dequantize, we divide $\operatorname{Q}(x)$ by $\Delta_q$, i.e., $\operatorname{Q}^{-1}(x) = \frac{x}{\Delta_q}$.
Note that our approach focuses primarily on secure aggregation and is independent of the specific quantization scheme and is compatible with other quantization methods, such as in \cite{zhang2020batchcrypt}.

Since our method focuses on secure aggregation rather than quantization, following most of the existing work \cite{bonawitz2017practical, karthikeyan2025opa}, we assume the inputs are integer-valued vectors. The user weights can be multiplied into the inputs during encoding. Thus, the aggregation formula typically considered in secure aggregation is:

\begin{equation}
\mathbf{x}_0^{(r)} = \sum_{i=1}^N \operatorname{Q}(\Vert \mathcal{D}_i \Vert \times \mathbf{x}_i^{(r)} ).
\end{equation}

Before training, the user decodes the global aggregation result to obtain the global model parameters:

\begin{equation}
\mathbf{w}_0^{(r)} = \frac{\operatorname{Q}^{-1}(\mathbf{x}_0^{(r)})}{\sum_{i=1}^N \Vert \mathcal{D}_i \Vert}.
\end{equation}
The division by $\sum_{i=1}^N \Vert \mathcal{D}_i \Vert$ after dequantization is necessary because, during each aggregation, some users may drop out. Hence, the users participating in the aggregation are not known before the aggregation is completed.

The dataset size can be treated as a public parameter or appended to the local model parameters (i.e., by setting $\mathbf{x}_i^{(r)}[d+1] = \Vert \mathcal{D}_i \Vert$) for private weight summation. Therefore, in the following discussion, we assume that the input $\mathbf{x}_i^{(r)}$ for each user $U_i$ in the $r$-th round of secure aggregation is a $d$-dimensional vector in the finite ring $\mathbb{Z}_{p_m}$. The goal of secure aggregation is to securely compute:
\begin{equation}
\mathbf{x}_0^{(r)} = \sum_{i \in \mathcal{U}} \mathbf{x}_i^{(r)} \bmod p_m,
\end{equation}
where $\mathcal{U}$ denotes the set of users online in the $r$-th round.

\subsection{Secret Sharing}
\label{apx:ss}
In a Secret Sharing (SS) scheme, there exists a dealer and $M$ share holders, denoted as $P_1, P_2, \cdots, P_M$. The dealer divides the secret $s$ into $M$ shares, $s_1, s_2, \cdots, s_M$, using a sharing algorithm, and transmits each share $s_m$ to the corresponding holder $P_m$ via a secure communication channel. The security of the sharing process ensures that no single share $s_m$ reveals any information about the secret $s$. 
When reconstruction is needed, each holder $P_m$ sends their share $s_m$ to the reconstruction entity (the server). The server then uses a reconstruction algorithm along with the shares $s_1, s_2, \cdots, s_M$ to recover the secret $s$.
In the following, we describe two SS schemes used in this work: the 2-out-of-2 additive SS and Shamir's SS.

\subsubsection{Additive Secret Sharing} 
In this work, we consider the 2-out-of-2 scheme, where there are exactly two share holders. We use $s^{A_1}$ and $s^{A_2}$ to denote the two share values, respectively.
At the sharing process, the dealer randomly samples a value $s^{A_1}$ from $\mathbb{Z}_p$, and then computes:
\begin{equation}
    s^{A_2} = s - s^{A_1} \bmod{p}.
\end{equation}
The reconstruction algorithm is simple and requires only one modular addition:
\begin{equation}
    s = s^{A_1} + s^{A_2} \bmod{p}.
\end{equation}
Since $s^{A_1}$ is uniformly sampled, the value of $s^{A_2}$ is also uniformly distributed over $\mathbb{Z}_p$. Therefore, knowing either $s^{A_1}$ or $s^{A_2}$ alone reveals no information about $s$.

\subsubsection{Shamir's Secret Sharing}
A Shamir's SS is a $(T, M)$ threshold scheme. Any $T-1$ shares provide no information about the secret $s$. However, any $T$ shares can be used to reconstruct the secret $s$. 

In the sharing phase, the dealer first randomly samples $T-1$ values $c_1, c_2, \cdots, c_{T-1}$ from $\mathbb{Z}_p$, and then constructs a polynomial of degree $T-1$:
\begin{equation*}
    f(x) = s + c_1 x + c_2 x^2 + \cdots + c_{T-1} x^{T-1}.
\end{equation*}
The share of any holder $P_m$ is $s_m = f(\beta_m)$, where $\beta_m$ is invertible in $\mathbb{Z}_p$ and satisfies the condition that $\beta_i \neq \beta_j$ for all $i \neq j$.

The reconstruction of the secret is straightforward. We can reconstruct the unique polynomial $f(x)$ with degree $T-1$ from any $T$ shares, for example $\{(\beta_{1}, s_{1}), (\beta_{2}, s_{2}), \cdots, (\beta_{T}, s_{T})\}$ using polynomial interpolation. The most commonly used interpolation method is Lagrange interpolation. First, we compute the Lagrange basis functions:
\begin{equation*}
    l_{i}(x) = \left(\prod_{j=1, j \neq i}^{T} (x - \beta_{j})\right) \times \operatorname{Inv}\left(\prod_{j=1, j \neq i}^{T} (\beta_{i} - \beta_{j})\right),
\end{equation*}
where $\operatorname{Inv}(x)$ denotes the inverse of $x$ in $\mathbb{Z}_p$, i.e., $x \times \operatorname{Inv}(x) \equiv 1 \bmod{p}$, which can be computed using the extended Euclidean algorithm.
Finally, the secret is reconstructed using the Lagrange basis functions as follows:
\begin{equation}
\label{eq:shamir_rec}
    s = \sum_{i=1}^T s_{i} \times l_{i}(0) \bmod{p}.
\end{equation}

\section{Proofs}
\subsection{Proof of \cref{thm:ss_security_distinguish}}
\label{apx:proof_ss_security_distinguish}
\begin{proof}
    We use a simulation-based approach \cite{lindell2017simulate} to prove \cref{thm:ss_security_distinguish}, which involves constructing a polynomial-time simulator that is independent of the input. The output of this simulator is computationally indistinguishable from the adversary's view during the execution of the real protocol. If such a simulator exists, it is evident that the protocol execution does not leak any information about the input to the adversary.
    Below, we construct a simulator whose output is indistinguishable from $\{s_m^{A_2}\}_{m \in [1:M]}^{x_1}$.

    The shares $\{s_m\}_{m \in [1:M]}^{x_1}$ are the share values of Shamir's SS, so they are computationally indistinguishable from a random distribution over $\mathbb{Z}_p$. Thus, the simulator can use a random input $r$ in place of $x_1$ to execute the Shamir's SS sharing algorithm, resulting in $\{s_m\}_{m \in [1:M]}^{r}$, which is indistinguishable from the original shares $\{s_m\}_{m \in [1:M]}^{x_1}$. 
    
    Notice that in each sharing instance, either $\kappa_{d,m}$ is different (due to re-executing the Setup), or $v_m$ is different (for different inputs), so each time $\mathcal{F}_{\kappa_{d,m}}(v_m, 1)$ generates an output that is equivalent to resampling from $\mathbb{Z}_p$. Since the output of $\mathcal{F}_{\kappa_{d,m}}(v_m, 1)$ is pseudorandom, it is computationally indistinguishable from the uniform distribution over $\mathbb{Z}_p$. Therefore, we directly use uniform sampling from $\mathbb{Z}_p$ to replace $s_m^{A_1}$ in the computation of $s_m^{A_2}$. The resulting output $\{s_m^{A_2}\}_{m \in [1:M]}^{r}$ is indistinguishable from $\{s_m^{A_2}\}_{m \in [1:M]}^{x_1}$.

    Based on the construction of the simulator, we know that the simulator's output is independent of the secret input, thus proving \cref{thm:ss_security_distinguish}.
\end{proof}

\subsection{Proof of \cref{thm:security_nfsa}}
\label{apx:proof_of_security_nfsa}
\begin{proof}
    Since our scheme does not require the recovery of a user's private key when they drop out, the server can only access less information about an offline user compared to an online user. If the privacy of online users is not compromised, it is clear that the privacy of offline users is also protected. Therefore, for the sake of simplicity in the proof, we assume that no users drop out.

    In the offline phase in \cref{alg:offline}, the server obtains the public keys of all users, $\{pk_i^u\}_{i \in [1:N]}$, and the public keys of the decryptors, $\{pk_m^p\}_{m \in [1:M]}$. In the $r$-th round of secure aggregation in \cref{alg:online}, the messages sent by users to the server consist of two parts. The first part is the share of the KhPRF keys: $\{\{\mathbf{s}_m^{A_2,i,r}\}_{m \in [1:M]}\}_{i \in \mathcal{U}}$, and the second part is the ciphertext of the user inputs: $\{\mathbf{c}_i^{(r)}\}_{i \in \mathcal{U}}$. 
    In the masking step, the server receives the share values $\{\mathbf{s}_m^{A_1,r}\}_{m \in \mathcal{P}_2}$ from the decryptors. Therefore, the real view $V_{real}$ of the server during the execution of the \NFSA protocol is composed of the initialized public keys and the shares and ciphertexts from each round, namely $\{\{\mathbf{s}_m^{A_2,i,r}\}_{m \in [1:M]}\}_{i \in \mathcal{U}}$, $\{\mathbf{c}_i^{(r)}\}_{i \in \mathcal{U}}$, and $\{\mathbf{s}_m^{A_1,r}\}_{m \in \mathcal{P}_2}$.

    We construct a polynomial-time simulator using the hybrid argument method \cite{lindell2017simulate} to simulate the server's view. Specifically, we begin with the real-world setting and construct a sequence of simulators. By making finite, polynomial-time modifications, we ensure that the views obtained by the server from interactions with any two adjacent simulators are indistinguishable. The behavior of the final simulator is independent of the private inputs and can be executed in the ideal world (where no user privacy data is involved), thereby completing the construction of the required simulator. The detailed construction is as follows:

    \begin{itemize}
        \item[$\mathtt{S}_0$] In this simulator, we make no modifications, so its behavior is identical to that of the users and decryptors in the real protocol, generating the real view $V_{real}$.
        
        \item[$\mathtt{S}_1$] In this simulator, we modify the shared PRF key of user $U_i$ in \cref{alg:online}. Instead of using $\kappa_{i,m}^{(r)}$, we replace it with a randomly generated key of the same length, denoted $\kappa_{i,m}^{(r),r}$. The decryptor no longer derives the key of user $U_i$ through the hash function, but instead uses $\kappa_{i,m}^{(r),r}$ directly. The other steps remain unchanged. We note that, by the security of the hash function, $\kappa_{i,m}^{(r)}$ and $\kappa_{i,m}^{(r),r}$ are indistinguishable. This means we are simply using a new key for \TLSS, and according to \cref{thm:ss_security_distinguish}, the server cannot distinguish between different \TLSS shares. Therefore, the generated server's views by $\mathtt{S}_1$ and $\mathtt{S}_0$ are indistinguishable.
        
        \item[$\mathtt{S}_2$] 
        In this simulator, all the PRF keys required by \TLSS for each user in every round are replaced with random keys, while the other steps remain the same as in $\mathtt{S}_1$. Based on the analysis in $\mathtt{S}_1$, the generated server's views by $\mathtt{S}_2$ and $\mathtt{S}_1$ are indistinguishable.
        
        \item[$\mathtt{S}_3$] 
        In this simulator, the public keys that need to be sent to the server in \cref{alg:offline}, including all users' public keys $\{pk_i^u\}_{i \in [1:N]}$ and the decryptors' public keys $\{pk_m^p\}_{m\in[1:M]}$, are replaced with randomly generated public keys and sent to the server. The other steps remain the same as in $\mathtt{S}_2$. Note that the online and offline phases of $\mathtt{S}_2$ are independent, and the online phase does not rely on the keys negotiated in the offline phase of $\mathtt{S}_2$. Therefore, the online phase of $\mathtt{S}_3$ is identical to that of $\mathtt{S}_2$. Due to the security of KA, the randomly generated public keys in $\mathtt{S}_3$ are indistinguishable from the public keys generated by KA in $\mathtt{S}_2$. Hence, the generated server's views by $\mathtt{S}_3$ and $\mathtt{S}_2$ are indistinguishable.

        \item[$\mathtt{S}_4$] 
        In this simulator, in each round, one online user $U_i^{(r)}$ is selected, and their input is replaced with $\mathbf{x}^{(r)}$, while the inputs of all other users are replaced with zero vectors of the same length. The other steps remain the same as in $\mathtt{S}_3$. Note that, after this replacement, the final result is still $\mathbf{x}^{(r)}$.
        Due to the security of KhPRF, $F_H(\mathbf{k}_i^{(r)}, r)$ generated in \cref{alg:online} and the random value are indistinguishable, so $\mathbf{c}_i^{(r)}$ and the random value are also indistinguishable. Since the sharing part of $\mathtt{S}_4$ is identical to that of $\mathtt{S}_3$, the generated server's views by $\mathtt{S}_4$ and $\mathtt{S}_3$ are indistinguishable.

    \end{itemize}

    The simulator $\mathtt{S}_4$ is a polynomial-time simulator independent of the users' inputs and internal states. Due to the transitivity of indistinguishability, the server's view generated by the interaction between $\mathtt{S}_4$ and the server is indistinguishable from the real server's view in the actual protocol execution. Thus, we have completed the proof of \cref{thm:security_nfsa}.
    
\end{proof}

\section{KhPRF based on RLWR}
\label{apx:KhPRF_rlwr}

When long pseudo-random numbers need to be generated using KhPRF, the efficiency of implementations based on LWR is far inferior to those based on RLWR. LWR requires matrix multiplication, while RLWR operates in the polynomial ring. By using NTT (Number Theoretic Transform) representation, fast polynomial multiplication can be achieved. Below, we describe a variant of KhPRF based on RLWR and provide some of our optimization techniques.

Let the polynomial ring be defined as $\mathcal{R}_q = \mathbb{Z}_q[x]/p(x)$, where $p(x)$ is a degree-$n_\lambda$ cyclotomic polynomial. Let $H_3(x): \chi \to \mathcal{R}_q$ represent a hash function mapping the input $x$ to an element of $\mathcal{R}_q$. In the RLWR variant, the construction of the KhPRF in \cref{eq:LWR_random_oracle} is modified as follows:
\begin{equation} 
\label{eq:RLWR_khprf} 
F_H(\mathbf{k}, x) = \lfloor H_3(x) \times \mathbf{k} \rfloor_{p_r},
\end{equation}
where the multiplication is performed in the polynomial ring, and $\lfloor \mathbf{x} \rfloor_{p_r}$ denotes rounding each coefficient of $\mathbf{x}$ (in the power basis).

A polynomial from the ring has $n_\lambda$ coefficients, so a single computation can yield $n_\lambda$ pseudo-random numbers. If the NTT transformation is used for computation, the complexity is $O(n_\lambda\log n_\lambda)$ \cite{mert2019design}. On the other hand, directly using $n_\lambda$ invocations of \cref{eq:LWR_random_oracle} to generate $n_\lambda$ pseudo-random numbers results in a computational complexity of $O(n_\lambda^2)$.

Typically, we set $p(x) = x^{n_\lambda} + 1$, where $n_\lambda$ is a power of 2 and $q \equiv 1 \bmod 2n_\lambda$. According to the properties of cyclotomic polynomials, $p(x)$ has $n_\lambda$ roots, and their order is $2n_\lambda$. Thus, it can be represented in the form of Double CRT or as an NTT expression \cite{gentry2012homomorphic}. By applying the NTT transformation, we convert the coefficient representation of the polynomial into its NTT representation. This transformation is bijective. In the NTT representation, multiplication in the polynomial ring is essentially the element-wise multiplication of two vectors, requiring only $O(n_\lambda)$ complexity. The time complexity of both the NTT forward and inverse transforms is $O(n_\lambda \log n_\lambda)$.

\subsection{Our Optimization}
We observe that the bijection implies that if the coefficients of an element are uniformly distributed in the coefficient representation, they will also be uniformly distributed in the NTT representation. Therefore, the key $\mathbf{k}$ for the KhPRF can be directly sampled in the NTT representation. Similarly, $H_3(x)$ can directly map $x$ to an element of $\mathcal{R}_q$ in its NTT representation. As a result, the NTT transformation is not required, and the multiplication in \cref{eq:RLWR_khprf} only incurs $O(n_\lambda)$ complexity. However, the rounding operation is not linear and cannot be directly performed in the NTT representation. Thus, after obtaining $H_3(x) \times \mathbf{k}$, an inverse NTT transformation is required to perform the rounding. Therefore, the final computation for the KhPRF is:

\begin{equation}
    \label{eq:RLWR_khprf_ntt}
    F_H(\mathbf{k},x) = \lfloor \operatorname{NTT}^{-1}(H_3(x) \cdot \mathbf{k} \bmod q) \rfloor_{p_r},
\end{equation}
where $\cdot$ represents the element-wise product modulo $q$.
Since the NTT transformation is linear, it does not affect the homomorphic properties of the KhPRF.

\subsection{Parameter Selection for KhPRF}
\label{apx:parameter_select_KhPRF}
One potential difficulty in implementation is finding the roots of $p(x)$ in $\mathbb{Z}_q$. We note that $q \equiv 1 \mod 2n_\lambda$, so there exists an integer $t = \frac{q-1}{2n_\lambda}$. First, we obtain a primitive root $g$ of $\mathbb{Z}_q$ using a random algorithm. Then, we compute $w = g^t \mod q$. This $w$ is the root we need, and the other roots are given by $w^{2i+1}$ where $i = 1, 2, \dots, n_\lambda-1$. The following provides a simple proof:

Since $t = \frac{q-1}{2n_\lambda}$, we have:
\begin{equation*}
    \begin{aligned}
        w^{n_\lambda}&\equiv g^{tn_\lambda} \mod q\\
        &\equiv g^{\frac{q-1}{2n_\lambda}\times n_\lambda} \mod q\\
        & \equiv g^{\frac{q-1}{2}} \mod q.
    \end{aligned}
\end{equation*}
Since $g$ is a primitive root, we know that $g^{\frac{q-1}{2}} \equiv -1 \mod q$. Therefore, we have $w^{n_\lambda} \equiv -1 \mod q$. Additionally, $(w^{2i+1})^{n_\lambda} \equiv (w^{2n_\lambda})^i \times w^{n_\lambda} \equiv -1 \mod q$. Since $w^{2i+1}$ for $i = 1, 2, \dots, n_\lambda-1$ are distinct, and $p(x)$ has only $n_\lambda$ roots, we have found all the roots.

Note that if $w^{2i+1} \equiv w^{2j+1} \mod q$ when $i > j$ and $i, j \in [1:n_\lambda-1]$, then $w^{2(i-j)} \equiv 1 \mod q$, which contradicts the fact that the order of $w$ is $2n_\lambda$. This is why we require $q \equiv 1 \mod 2n_\lambda$. Since $n_\lambda$ is a power of 2, we can leverage the NTT algorithm to achieve a computational complexity of $O(n_\lambda \log n_\lambda)$.

Next, we discuss how to choose the security parameters $\Delta, p_c, q$, and $n_\lambda$ for the KhPRF.
The rounding parameter of LWR for KhPRF is $p_r=\Delta p_c$.
According to the results in OPA \cite{karthikeyan2025opa}, given an LWR instance $(\mathbf{A}, \mathbf{t} = \lfloor \mathbf{A} \mathbf{k} \rfloor_{p_r})$, where $\mathbf{A} \in \mathbb{Z}_q^{L \times n_\lambda}$ and $\mathbf{k} \in \mathbb{Z}_q^{n_\lambda}$, we can convert it into an LWE instance as $(\mathbf{A}, \mathbf{A} \mathbf{k} + \mathbf{e})$, where $\mathbf{e} = \frac{q}{p_r} \times \mathbf{t} - \mathbf{A} \mathbf{k}$.
The variance of the error in the LWE samples is $\frac{q^2}{12 p_r^2}$. Based on this noise level, we use the LWE estimator to determine the value of $n_\lambda$ and the corresponding maximum $q_{\max}$ for a given security parameter. For the KhPRF based on RLWR, the process is similar; however, we must choose $q \equiv 1 \mod 2n_\lambda$, ensuring that $q < q_{\max}$ and that $q$ is a prime number.

\begin{figure}[!t]
	\centering  
    \subfloat[Holder Time]{
    \label{fig:m127_dix_holder_time}
    \includegraphics[width=0.45\linewidth]{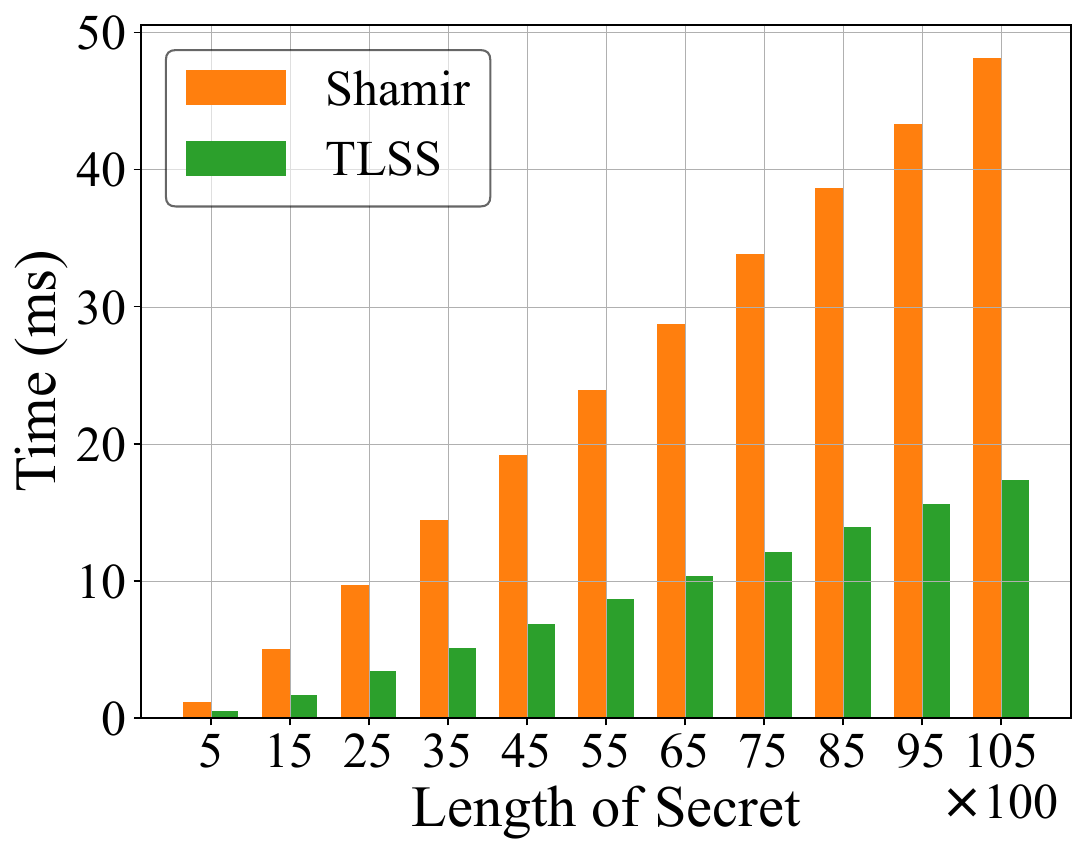}}
     \subfloat[Holder Traffic]{
    \label{fig:m127_dix_holder_traffic}
    \includegraphics[width=0.45\linewidth]{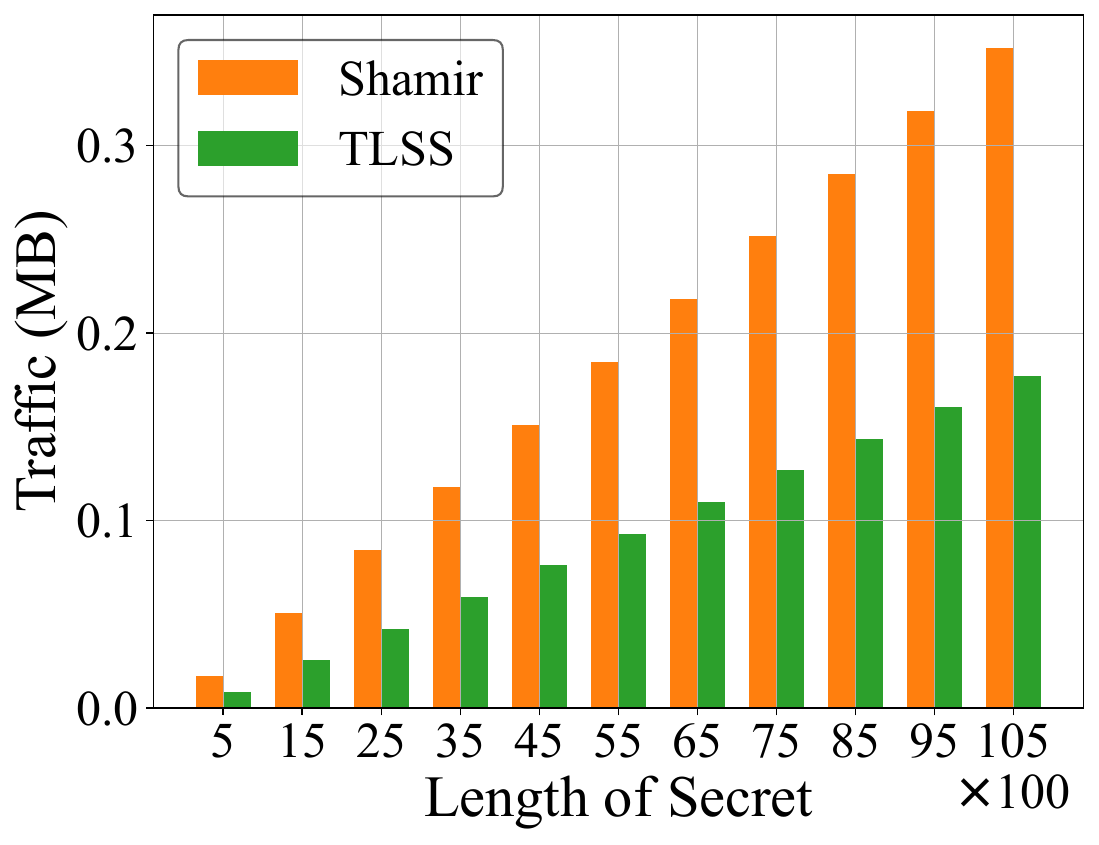}}
    
\caption{Holder performance with modulus bit length 128.}
\label{fig:m127_holder_performance}
\end{figure}

\begin{figure}[!t]
	\centering 
    \subfloat[Total Time]{
    \label{fig:m127_dix_sim_time}
    \includegraphics[width=0.45\linewidth]{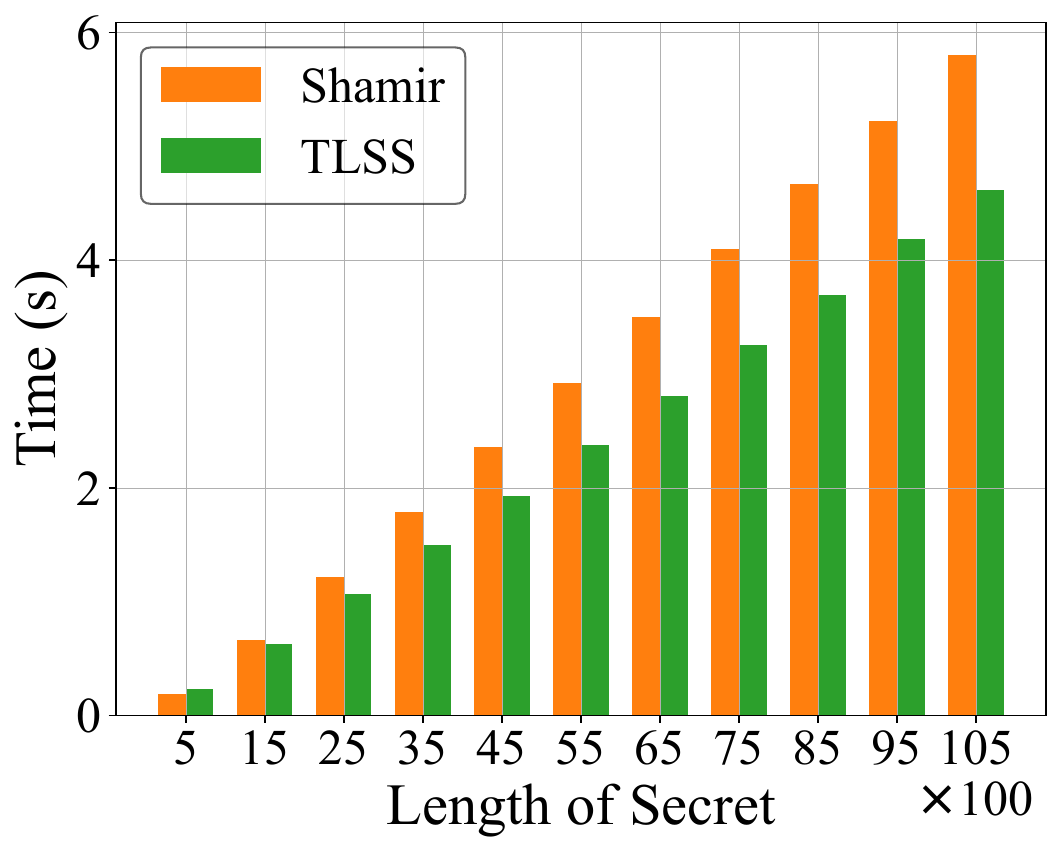}}
     \subfloat[Total Traffic]{
    \label{fig:m127_dix_server_traffic}
    \includegraphics[width=0.45\linewidth]{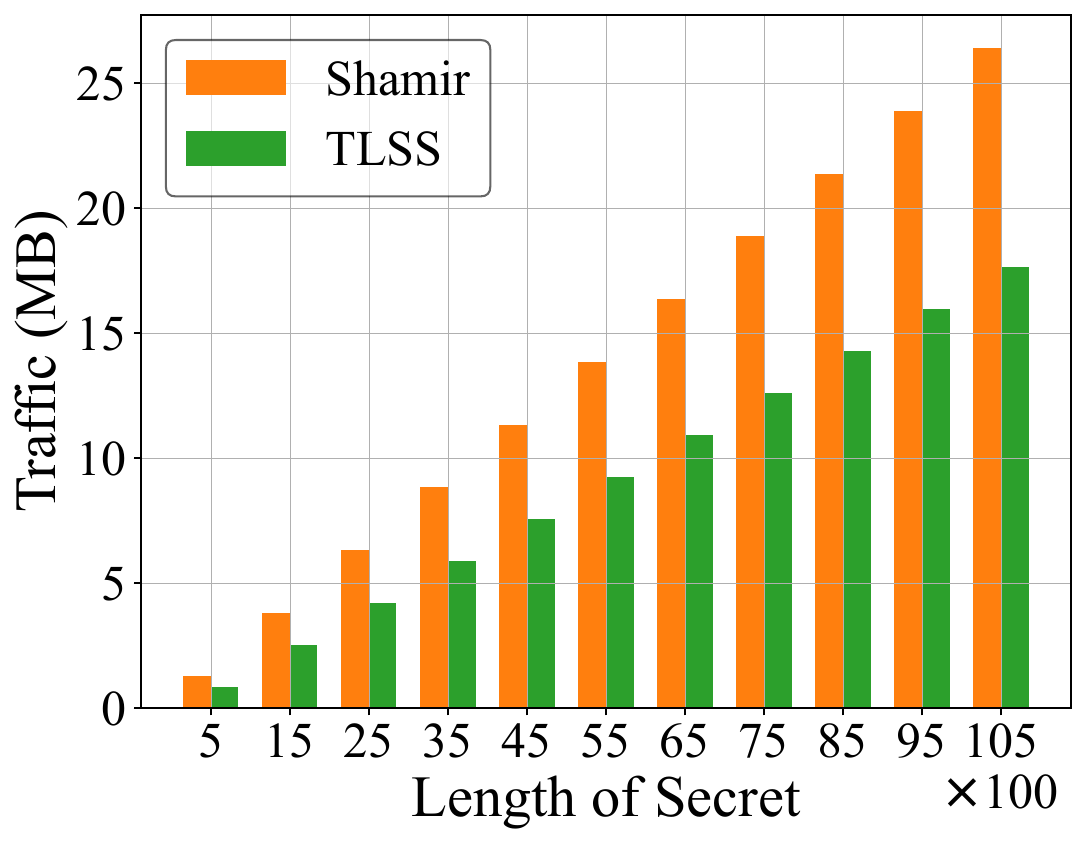}}
    
\caption{Comparison of total overhead under different secret lengths with holder number 50 and modulus bit length 128.
}
\label{fig:m127_dix_total_overhead}
\end{figure}

\begin{figure}[!t]
	\centering 
    \subfloat[Total Time]{
    \label{fig:m127_dih_sim_time}
    \includegraphics[width=0.45\linewidth]{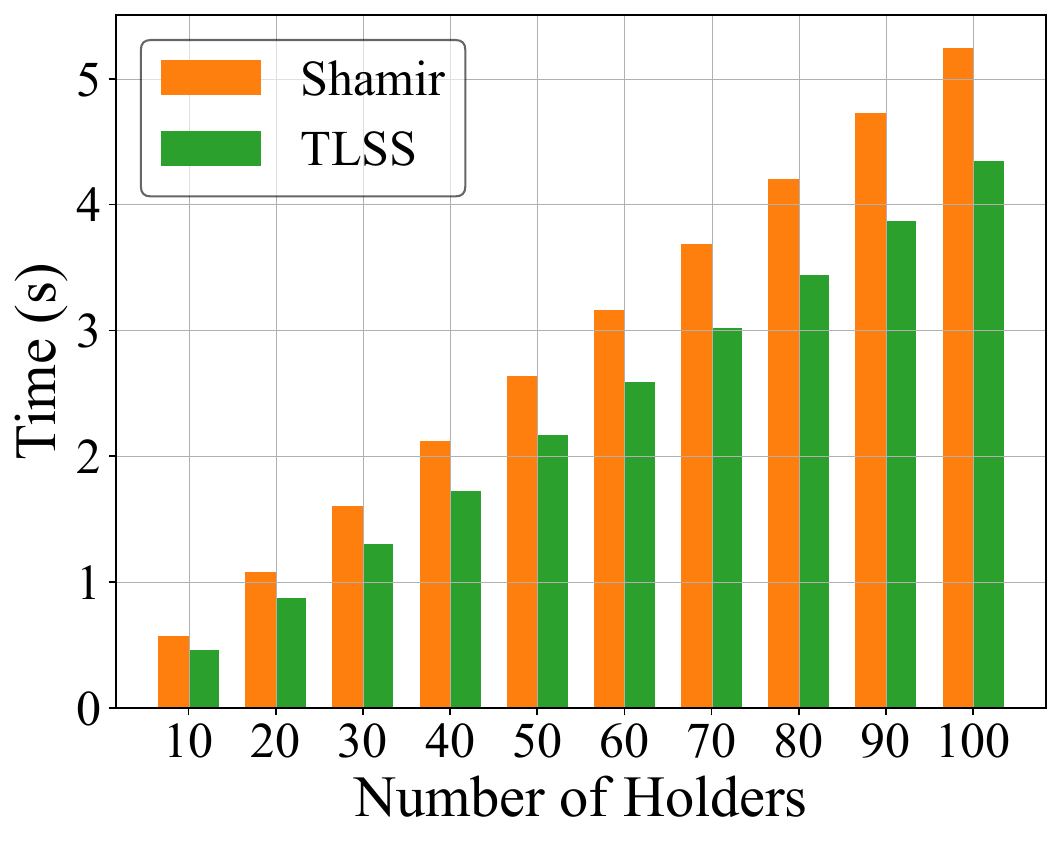}}
     \subfloat[Total Traffic]{
    \label{fig:m127_dih_server_traffic}
    \includegraphics[width=0.45\linewidth]{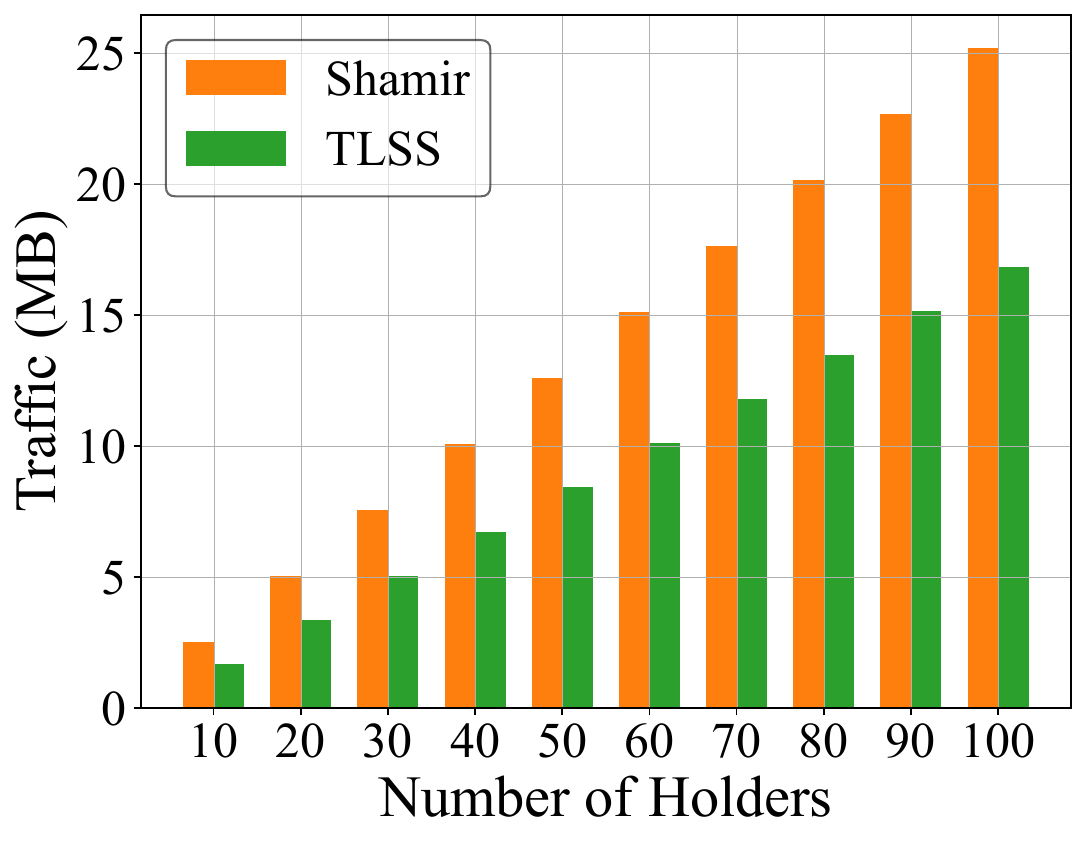}}
    
\caption{Comparison of total overhead under different numbers of holders with secret length 5000 and modulus bit length 128.}
\label{fig:m127_dih_overhead}
\end{figure}

\section{\TLSS Performance with Greater Modulus}
\label{apx:m128}

Since the bit length of the modulus affects the performance of \TLSS, we present a comparison of the overhead between \TLSS and Shamir when the bit length of the modulus is 128. Due to the inefficiency of our implementation, the time required for generating the PRF is significantly larger compared to when the modulus bit length is 64. However, the computation time for the holders, under different secret lengths, remains noticeably lower than Shamir, as shown in \cref{fig:m127_dix_holder_time}. The communication overhead for the holders is also much lower than Shamir, even more significantly so compared to when the modulus bit length is 64. This is because the ratio of communication overhead between \TLSS and Shamir does not vary significantly across different moduli, and the increase in modulus bit length to 128 doubles the overall overhead. Therefore, the difference in communication overhead for the holders in \TLSS becomes even more pronounced. Furthermore, the total overhead across different secret lengths and the total cost across different numbers of holders still shows that \TLSS significantly outperforms Shamir, as shown in \cref{fig:m127_dix_total_overhead} and \cref{fig:m127_dih_overhead}.

\end{document}